\newcommand{\longerversion}[1]{}
\newcommand{\longversion}[1]{#1}
\newcommand{\shortversion}[1]{}
\shortversion{
\documentclass{svproc}
  \usepackage[belowskip=-5pt,aboveskip=5pt]{caption} %
}

\longversion{
\documentclass[10pt,letter]{article}%
  \usepackage[hscale=0.9,scale=0.7]{geometry}
  \usepackage{caption}
}

\usepackage{url}  %
\usepackage{hyperref}

\usepackage{multirow}
\usepackage{amssymb,amsfonts,amsmath} 

\longversion{
	\usepackage{amsthm}%
}
\usepackage{stackengine}
\usepackage{calc}
\newlength\shlength
\newcommand\xshlongvec[2][0]{\setlength\shlength{#1pt}%
  \stackengine{-5.6pt}{$#2$}{\smash{$\kern\shlength%
    \stackengine{7.55pt}{$\mathchar"017E$}%
      {\rule{\widthof{$#2$}}{.57pt}\kern.4pt}{O}{r}{F}{F}{L}\kern-\shlength$}}%
      {O}{c}{F}{T}{S}}

\usepackage{xspace}
\usepackage[ruled,vlined,linesnumbered]{algorithm2e}
\usepackage[table,x11names]{xcolor}
\usepackage{graphicx}  %
\renewcommand*{\algorithmcfname}{Listing}
\SetKwInput{KwData}{In}
\SetKwInput{KwResult}{Out}
\SetAlFnt{\small}
\SetAlCapFnt{\small}
\SetAlCapNameFnt{\small}
\SetAlCapHSkip{0pt}
\SetEndCharOfAlgoLine{}
\def\etal{et~al.{}}
\newcommand{\mtext}[1]{\ensuremath{\mathcal{#1}}}
\newcommand{\cnt}[0]{\ensuremath{\#}}
\newcommand{\cntc}[0]{\ensuremath{\cnt\cdot}}

\DeclareMathOperator{\SP}{ProofStates}
\DeclareMathOperator{\checkord}{CheckOrder}
\DeclareMathOperator{\checkmod}{CheckMod}
\DeclareMathOperator{\gatherproof}{proven}

\DeclareMathOperator{\possord}{ords}
\newcommand{\MAI}[2]{\ensuremath{#1^+_{#2}}}%
\newcommand{\MAR}[2]{\ensuremath{#1^-_{#2}}}%
\newcommand{\MARR}[2]{\ensuremath{#1^\sim_{#2}}}%

\newcommand{\FIX}[1]{#1}

\DeclareMathOperator{\post}{post-order}
\usepackage{microtype}
\usepackage{tikz}
\usetikzlibrary{shapes}
\usetikzlibrary{calc}

\shortversion{
}

\renewcommand{\P}{\ensuremath{\textsc{P}}\xspace}
\newcommand{\NP}{\ensuremath{\textsc{NP}}\xspace}
\newcommand{\SIGMA}[2]{\ensuremath{\Sigma_{\textrm{#1}}^{\textrm{#2}}}}

\usetikzlibrary{positioning}

\tikzstyle{tdnode} = [draw,rounded corners,top color=vertexTopColor,bottom color=vertexBottomColor,minimum size=1.5em]
\tikzstyle{stdnode} = [tdnode, font=\scriptsize]
\tikzstyle{stdnodecompact} = [stdnode, inner sep = 1.5pt, outer sep = 0.1pt]
\tikzstyle{stdnodetable} = [stdnode, inner sep = 0.5pt, outer sep = 0]
\tikzstyle{stdnodenum} = [minimum size=1.5em, font=\scriptsize]
\tikzstyle{tdedge} = [-,draw,thick]
\tikzstyle{tdlabel} = [draw=none, rectangle, fill=none, inner sep=0pt, font=\scriptsize]
\colorlet{vertexTopColor}{white}
\colorlet{vertexBottomColor}{black!10}

\usepackage{ctable}
\usepackage{ifthen}
\newif\iflong
\usepackage{todonotes}
\newcommand{\restrict}[2]{\ensuremath{#1\cap #2}}

\usepackage{float}
\newcommand{\SB}{\{}%
\newcommand{\SM}{\mid}%
\newcommand{\SE}{\}}%
\def\hy{\hbox{-}\nobreak\hskip0pt}

\newcommand{\mdpa}[1]{\ensuremath{\mathtt{PCNT}_{#1}}}
\newcommand{\ta}[1]{\ensuremath{2^{#1}}}
\newcommand{\Card}[1]{\left|#1\right|}
\newcommand{\CCard}[1]{\|#1\|}
\newcommand{\algo}[1]{\ensuremath{\mathbb{#1}}}
\DeclareMathOperator{\width}{width}
\DeclareMathOperator{\children}{children}

\newcommand{\sharpCONP}[0]{\#\ensuremath{\cdot}co\textsc{NP}}
\newcommand{\sharpNP}[0]{\#\ensuremath{\cdot}\textsc{NP}}

\newcommand{\sharpSigma}[1]{\#\ensuremath{\cdot\Sigma_{#1}P}}

\makeatletter
\newcommand{\algorithmfootnote}[2][\footnotesize]{
  \let\old@algocf@finish\@algocf@finish
  \def\@algocf@finish{\old@algocf@finish
    \leavevmode\rlap{\begin{minipage}{\linewidth}
    #1#2
    \end{minipage}}
  }
}
\makeatother

\DeclareMathOperator{\orig}{\algo{A}\hy origins}
\DeclareMathOperator{\origs}{\algo{A}\hy origins}
\newcommand{\origa}[1]{\operatorname{#1\hy origins}}
\newcommand{\origse}[1]{\operatorname{#1\hy origins}}
\DeclareMathOperator{\Ext}{Ext}
\DeclareMathOperator{\Exts}{Exts}
\DeclareMathOperator{\PExt}{SatExt}
\DeclareMathOperator{\pmc}{pasc}
\DeclareMathOperator{\ipmc}{ipasc}
\DeclareMathOperator{\bucket}{=_P}%
\DeclareMathOperator{\buckets}{buckets}
\DeclareMathOperator{\subbuckets}{sub\hy buckets}

\newcommand{\TTT}{\ensuremath{\mathcal{T}}}%
\newcommand{\WWW}{\ensuremath{\mathcal{W}}}%
\newcommand{\por}{\vee}

\newcommand{\eqdef}{\ensuremath{\,\mathrel{\mathop:}=}}
\newcommand{\hsep}{\leftarrow\,}
\newcommand{\ASP}{\textsc{As}\xspace}
\newcommand{\SAT}{\textsc{SAT}\xspace}

\newcommand{\cSAT}{\cnt\textsc{SAT}\xspace}
\newcommand{\cASP}{\cnt\textsc{As}\xspace}
\newcommand{\PASP}{\cnt\textsc{PAs}\xspace}%
\newcommand{\PDASP}{\PASP}%
\newcommand{\AlgA}{\algo{A}}%
\newcommand{\AlgS}{\AlgA}%
\newcommand{\PROJ}{\algo{PROJ}\xspace}
\newcommand{\at}{\text{\normalfont at}}
\newcommand{\bigO}[1]{\ensuremath{{\mathcal O}(#1)}}
\newcommand{\CCC}{\ensuremath{\mathcal{C}}}%

\newcommand{\tuplecolor}[1]{\textcolor{#1}}
\newcommand{\inputPredColor}{orange!55!red}
\newcommand{\outputPredColor}{blue!45!black}
\newcommand{\statePredColor}{green!62!black}

\newcommand{\tabval}{\ensuremath{u}}
\newcommand{\tab}[1]{\ensuremath{\tau_{#1}}}

\newcommand{\att}[1]{\ensuremath{\at_{\hspace{-0.05em}\leq\hspace{-0.05em}#1}}}
\newcommand{\attneq}[1]{\ensuremath{\at_{\hspace{-0.05em}<\hspace{-0.05em}#1}}}

\newcommand{\prog}{\ensuremath{\Pi}}
\newcommand{\progt}[1]{\ensuremath{\prog_{\hspace{-0.05em}\leq\hspace{-0.05em}#1}}}
\newcommand{\progtneq}[1]{\ensuremath{\prog_{\hspace{-0.05em}<\hspace{-0.05em}#1}}}

\newcommand{\dpa}{\ensuremath{\mathtt{DP}}}
\newcommand{\Tab}[1]{\ensuremath{\text{C-Tabs}}}

\def\thyph{\text{-}\penalty0\hskip0pt\relax}

\newcommand{\ATab}[1]{\ensuremath{#1\thyph\text{Comp}}}

\newcommand{\tw}[1]{\mathit{tw}(#1)}

\newcommand{\Nat}{\mathbb{N}} %
\DeclareMathOperator{\type}{type}
\newcommand{\intr}{\textit{int}}
\newcommand{\leaf}{\textit{leaf}}
\newcommand{\rem}{\textit{rem}}
\newcommand{\join}{\textit{join}}

\newtheorem{observation}{Observation}
\longversion{
\newtheorem{example}{Example}
\newtheorem{conjecture}{Conjecture}
\newtheorem{proposition}{Proposition}

\newtheorem{theorem}{Theorem}
\newtheorem{lemma}{Lemma}
\newtheorem{definition}{Definition}
\newtheorem{corollary}{Corollary}}

{%
  \addtocounter{observation}{-1}
  \endgroup
}%

{%
  \addtocounter{corollary}{-1}
  \endgroup
}%

\newenvironment{restatelemma}[1][\unskip]{%
  \begingroup
  \def\thelemma{\ref{#1}} 
}%
{%
  \addtocounter{lemma}{-1}
  \endgroup
}%

\newenvironment{restateproposition}[1][\unskip]{%
  \begingroup
  \def\theproposition{\ref{#1}} 
}%
{%
  \addtocounter{proposition}{-1}
  \endgroup
}%

\newenvironment{restatetheorem}[1][\unskip]{%
  \begingroup
  \def\thetheorem{\ref{#1}} 
}%
{%
  \addtocounter{theorem}{-1}
  \endgroup
}%

\DeclareMathOperator{\pcnt}{pasc}
\DeclareMathOperator{\local}{local}
\DeclareMathOperator{\sipmc}{s-ipasc}
\DeclareMathOperator{\icnt}{ipasc}

\newcommand{\PRIM}{\ensuremath{{\algo{PHC}}}\xspace}
\newcommand{\INC}{\ensuremath{{\algo{INC}}}\xspace}

\longversion{
\urldef{\mailsa}\path|hecher@dbai.tuwien.ac.at|
\urldef{\mailsb}\path|johannes.fichte@tu-dresden.de|

\usepackage[numbers,sort&compress]{natbib}
\bibliographystyle{abbrvnat}
}
\shortversion{
\bibliographystyle{splncs03}
\renewenvironment{example}{\begin{EXa}}{\hfill\ensuremath{\blacksquare}\end{EXa}}
  \spnewtheorem{EXa}{Example}{\bfseries}{\normalfont}

}
 
\begin{document}
\shortversion{\mainmatter}              %
 \title{Treewidth and Counting Projected Answer Sets%
\thanks{This work extends an abstract~\cite{FichteHecher18c} explaining only concepts, and a preliminary workshop paper\longversion{~\cite{FichteHecher18}}, and %
     has been supported by FWF Grant Y698 and DFG Grant HO 1294/11-1. 
     The second author is also affiliated with University of
     Potsdam,~Germany. \longversion{The final publication will be available at Springer proceedings of LPNMR 2019.} 
   }
 }
 \longversion{
 \author{Johannes K. Fichte$^1$%
\and Markus Hecher$^2$\\%
\longversion{\\[3pt]
    $^1$: TU Dresden, Germany, \mailsb\\
    $^2$: TU Wien, Austria, \mailsa\\}
}}
\shortversion{
 \titlerunning{Treewidth and Counting Projected Answer Sets}
 \authorrunning{Fichte and Hecher}
 \tocauthor{Johannes K. Fichte, and Markus Hecher}
 \author{Johannes K. Fichte\inst{1} \and Markus Hecher\inst{2}}
 \institute{TU Dresden, Germany, %
   \email{johannes.fichte@tu-dresden.de} \and TU Wien, Austria, %
   \email{hecher@dbai.tuwien.ac.at}
 }}
\maketitle
\begin{abstract}
  In this paper, we introduce novel algorithms to solve
  \emph{projected answer set counting (\PASP)}. \PASP asks to count
  the number of answer sets with respect to a given set of
  \emph{projected atoms}, where multiple answer sets that are
  identical when restricted to the projected atoms count as only one
  projected answer set.
  Our algorithms exploit small treewidth of the primal graph of the
  input instance by dynamic programming (DP).

  We establish a new algorithm for head-cycle-free (HCF) programs and lift
  very recent results from projected model counting to \PASP when the
  input is restricted to HCF programs.
  Further, we show how established DP algorithms for tight, normal, and disjunctive
  answer set programs can be extended to solve \PASP.
  \FIX{Our algorithms run in polynomial time while requiring double exponential time in the treewidth 
  for tight, normal, and HCF programs, and triple exponential time for disjunctive
  programs.} %

  Finally, we take the exponential time hypothesis (ETH) into account
  and establish lower bounds of bounded treewidth algorithms for
  \PASP. \FIX{Under ETH, one cannot %
  significantly improve our obtained worst-case runtimes.}
\end{abstract}

\section{Introduction}
Answer Set Programming (ASP)~\cite{BrewkaEiterTruszczynski11} is an
active research area of artificial intelligence. It provides a
logic-based declarative modelling language and problem solving
framework\longversion{~\cite{GebserKaminskiKaufmannSchaub12}} for hard
computational problems\longversion{, which has been widely
  applied~\cite{BalducciniGelfondNogueira06a,NiemelaSimonsSoininen99,NogueiraBalducciniGelfond01a,GuziolowskiEtAl13a}}.
In ASP, questions are encoded into rules and constraints that form a
disjunctive (logic) program over atoms. Solutions to the program are
so-called answer sets.
Lately, two computational problems of ASP have received increasing
attention, namely, \cASP~\cite{FichteEtAl17a} and
\PASP~\cite{Aziz15a}.
The problem \cASP asks to \emph{output the number of answer sets} of a given
disjunctive program. When considering computational complexity \cASP
can be classified as \sharpCONP-complete~\cite{FichteEtAl17a}, which
is even harder than counting the models of a Boolean formula.
A natural abstraction of \cASP is to consider projected counting where
we ask to count the answer sets of a disjunctive program with respect
to a given set of \emph{projected atoms} (\PASP).  Particularly, %
\FIX{multiple answer sets that are identical when reduced to the projected
atoms are considered as only one solution}. \FIX{Intuitively, \PASP is needed to count answer sets without counting functionally independent auxiliary atoms.}
Under standard assumptions the problem~\PASP is complete for the class
\sharpSigma{2}.
However, if we take all atoms as projected, then \PASP is again
\sharpCONP-complete and if there are no projected atoms then it is
simply $\Sigma^p_2$-complete.
But some fragments of ASP have lower complexity.  A prominent example
is the class of \emph{head-cycle-free (HCF)}
programs~\cite{Ben-EliyahuDechter94}, which requires the absence of
cycles in a certain graph representation of the program. \FIX{
Deciding whether a HCF program has an answer~set%
~is~%
\NP-complete.}

A way to solve computationally hard problems is to employ
parameterized algorithmics~\cite{CyganEtAl15}, which exploits certain
structural restrictions in a given input instance. Because structural
properties of an input instance often allow for algorithms that solve
problems in polynomial time \FIX{in the} size of the input and exponential
time in a measure of the structure, whereas under standard assumptions
an efficient algorithm is not possible if we consider only the size of
the input.
In this paper, we consider the treewidth of a graph representation
associated with the given input program as structural restriction,
namely the \emph{treewidth of the primal
  graph}~\cite{JaklPichlerWoltran09}.
Generally speaking, treewidth\longversion{\footnote{Google Scholar outputs 18,800
  results employing treewidth (queried: March.\ 27, 2019).}} measures the closeness of a graph to a
tree, based on the observation that problems on trees are often easier
to solve than on arbitrary graphs.
Our results are as follows: 
\FIX{We establish the classical complexity of \PASP and a novel
  algorithm} that solves ASP problems by exploiting treewidth when the
input program is restricted to HCF programs in runtime single
exponential in the treewidth. %
We introduce a framework for counting projected answer sets by
exploiting treewidth. Therefore, we lift recent results from projected
model counting in the domain of Boolean formulas to counting projected
answer sets. We establish algorithms that are (i)~double exponential
in the treewidth if the input is restricted to tight, normal or HCF
programs and (ii)~triple exponential in the treewidth if we allow
disjunctive programs.
Using the exponential time hypothesis (ETH), we establish that
\PASP %
\emph{cannot} \FIX{be solved in time better than double exponential in
  the treewidth for tight, normal, and HCF programs, and not better
  than triple exponential for disjunctive programs,
  respectively.} %

\smallskip\noindent\textbf{Related Work.}
Gebser, Kaufmann and Schaub~\cite{GebserKaufmannSchaub09a} 
considered projected enumeration for
ASP. Aziz~\cite{Aziz15a} introduced techniques to modify modern
\longversion{ASP-solvers  in order to}\shortversion{solvers to} count projected answer sets.
\longversion{%
Jakl, Pichler and Woltran~\cite{JaklPichlerWoltran09} presented DP algorithms that solve ASP
counting in time double exponential in the treewidth.
Pichler~\etal~\cite{PichlerEtAl14} investigated the complexity of extended programs
and also presented DP algorithms for it. We employ ideas from their
algorithms to solve head-cycle-free programs.
}%
Fichte~\etal~\cite{FichteEtAl17a\longversion{,FichteEtAl17b}} %
presented algorithms \FIX{to solve~\cASP} for the full standard
syntax of modern ASP solvers. %
Recently, Fichte~\etal~\cite{FichteEtAl18} gave DP algorithms for projected \cSAT
including lower bounds\FIX{, c.f., Table~\ref{tbl:summary}}.

\vspace{-.7em}
\section{Preliminaries}
\vspace{-.5em}
\textbf{Basics and Combinatorics.} \longversion{For a set~$X$, let
  $\ta{X}$ be the \emph{power set of~$X$} consisting of all
  subsets~$Y$ with $\emptyset \subseteq Y \subseteq X$.}
For given sequence~$\vec s$ and integer~$i>0$, $\vec s_{(i)}$ refers
to the $i$-th element of~$\vec s$
and~$<_{\vec s} \eqdef \SB (\vec s_{(i)},\vec s_{(j)}) \SM 1 \leq i <
j \leq \Card{\vec s}\SE$ denotes its \emph{induced ordering}.
Given \longversion{some integer~$n$ and a family of }finite
sets~$X_1$, $X_2$, $\ldots$, $X_n$\FIX{, the generalized} %
\emph{inclusion-exclusion
  principle}\longversion{~\cite{GrahamGrotschelLovasz95a}} states
that\longversion{ the number of elements in the union over all subsets
  is}
$\Card{\cup^n_{j = 1} X_j} = \Sigma_{I \subseteq \{1, \ldots, n\}, I
  \neq \emptyset} (-1)^{\Card{I}-1} \Card{\cap_{i \in I} X_i}$.

\smallskip\noindent\textbf{Computational Complexity.}
\longversion{We assume familiarity with standard notions in
  computational complexity~\cite{Papadimitriou94} and parameterized
  complexity~\cite{CyganEtAl15}, and use counting complexity classes
  as defined by Durand, Hermann and
  Kolaitis~\cite{DurandHermannKolaitis05}.
}
\shortversion{%
  For parameterized complexity we refer to~\cite{CyganEtAl15} and for
  counting complexity classes to~\cite{DurandHermannKolaitis05}.
}%
Let $\Sigma$ and $\Sigma'$ be finite alphabets,
$I \in \Sigma^*$ \FIX{be an} \emph{instance} and $\CCard{I}$
\longversion{denote the \emph{size} of~$I$.}%
\FIX{\shortversion{denote its \emph{size}.}}
A \emph{witness function} %
 $\mathcal{W}\colon \Sigma^* \rightarrow 2^{{\Sigma'}^*}$ %
maps an instance~$I \in \Sigma^*$ to its \emph{witnesses}. 
A
\emph{parameterized counting
  problem}~$L: \Sigma^* \times \Nat \rightarrow \Nat_0$ is a
function that maps a given instance~$I \in \Sigma^*$ and an
integer~$k \in \Nat$ to the cardinality of its
witnesses~$\Card{\WWW(I)}$.
Let $\mtext{C}$ be a decision complexity class,~e.g., \P. Then,
$\cntc\mtext{C}$ denotes the class of all counting problems whose
witness function~$\WWW$ satisfies (i)~there is a
function~$f: \Nat_0 \rightarrow \Nat_0$ such that for every
instance~$I \in \Sigma^*$ and every $W \in \WWW(I)$ we have
$\Card{W} \leq f(\CCard{I})$ and $f$ is computable in
time~$\bigO{\CCard{I}^c}$ for some constant~$c$ and (ii)~for every
instance~$I \in \Sigma^*$ decision problem~$\WWW(I)$~is~in\longversion{the complexity class}~$\mtext{C}$.
\longversion{Then, $\cntc\P$ is the complexity class consisting of all counting
problems associated with decision problems in \NP.}
\smallskip\noindent\textbf{Answer Set Programming (ASP).}
We follow standard definitions of propositional disjunctive ASP. For
comprehensive foundations, we refer to introductory
literature~\cite{BrewkaEiterTruszczynski11\longversion{,JanhunenNiemela16a}}.
Let $\ell$, $m$, $n$ be non-negative integers such that
$\ell \leq m \leq n$, $a_1$, $\ldots$, $a_n$ be distinct
atoms. \longversion{Moreover, we}\shortversion{We} refer by \emph{literal} to an atom or the negation
thereof. %
A \emph{program}~$\prog$ is a \FIX{finite} set of \emph{rules} of the form
\(
a_1\por \cdots \por a_\ell \hsep a_{\ell+1}, \ldots, a_{m}, \neg
a_{m+1}, \ldots, \neg a_n.
\)
For a rule~$r$, we let $H_r \eqdef \{a_1, \ldots, a_\ell\}$,
$B^+_r \eqdef \{a_{\ell+1}, \ldots, a_{m}\}$, and
$B^-_r \eqdef \{a_{m+1}, \ldots, a_n\}$.
We denote the sets of \emph{atoms} occurring in a rule~$r$ or in a
program~$\prog$ by $\at(r) \eqdef H_r \cup B^+_r \cup B^-_r$ and
$\at(\prog)\eqdef \cup_{r\in\prog} \at(r)$.
Let $\prog$ be a program.
A program~$\prog'$ is a \emph{sub-program of~$\prog$} if~$\prog'\subseteq\prog$.
\longversion{Program~}$\prog$ is \emph{normal} if $\Card{H_r} \leq 1$ for
every~$r \in \prog$.
The \emph{positive dependency digraph}~$D_\prog$ of $\prog$ is the
directed graph defined on the set of atoms
from~$\bigcup_{r\in \prog}H_r \cup B^+_r$, where for every
rule~$r \in \prog$ two atoms $a\in B^+_r$ and~$b\in H_r$ are joined by
an edge~$(a,b)$.
A \emph{head-cycle} of~$D_\prog$ is an $\{a, b\}$-cycle\footnote{Let
  $G=(V,E)$ be a digraph and $W \subseteq V$. Then, a cycle in~$G$ is
  a $W$-cycle if it contains all vertices from~$W$.} for two distinct
atoms~$a$, $b \in H_r$ for some rule $r \in \prog$. 
\FIX{Program~$\prog$ is~\emph{tight} \FIX{(\emph{head-cycle-free}~\cite{Ben-EliyahuDechter94})} if~$D_\prog$ contains no cycle (head-cycle).}

An \emph{interpretation} $I$ is a set of atoms. $I$ \emph{satisfies} a
rule~$r$ if $(H_r\,\cup\, B^-_r) \,\cap\, I \neq \emptyset$ or
$B^+_r \setminus I \neq \emptyset$.  $I$ is a \emph{model} of $\prog$
if it satisfies all rules of~$\prog$, in symbols $I \models \prog$. %
The \emph{Gelfond-Lifschitz
  (GL) reduct} of~$\prog$ under~$I$ is the program~$\prog^I$ obtained
from $\prog$ by first removing all rules~$r$ with
$B^-_r\cap I\neq \emptyset$ and then removing all~$\neg z$ where
$z \in B^-_r$ from the remaining
rules~$r$~\cite{GelfondLifschitz91}. %
$I$ is an \emph{answer set} of a program~$\prog$ if $I$ is a minimal
model of~$\prog^I$. %
Deciding whether a disjunctive program has an answer set is
\SIGMA{2}{P}-complete~\cite{EiterGottlob95}. The problem is called
\emph{consistency} (\ASP) of an ASP program.
If the input is restricted to normal programs, the complexity drops to
\NP-complete~\cite{BidoitFroidevaux91\longversion{,MarekTruszczynski91}}.
A head-cycle-free program~$\prog$ %
can be translated into a normal program in polynomial
time~\cite{Ben-EliyahuDechter94}.
The following well-known characterization of answer sets is often
invoked when considering normal programs~\cite{LinZhao03}.
Given a model~$I$ of a normal program~$\prog$ and an ordering~$\sigma$ of
atoms over~$I$. An atom~$a\in I$ is \emph{proven} if there is a
rule~$r\in\prog$ with $a\in H_r$ where (i)~$B^+_r\subseteq I$,
(ii)~$b <_\sigma a$ for every~$b\in B_r^+$, and
(iii)~$I \cap B^-_r = \emptyset$ and
$I \cap (H_r \setminus \{a\}) = \emptyset$. Then, $I$ is an
\emph{answer set} of~$\prog$ if (i)~$I$ is a model of~$\prog$, and
(ii) every atom~$a \in I$ is proven.
This characterization vacuously extends to head-cycle-free
programs by applying the results of Ben-Eliyahu and Dechter~\cite{Ben-EliyahuDechter94}.
\FIX{Given a program~$\prog$, we assume %
in the following that every atom~$a\in\at(\prog)$ %
occurs in some head of a rule
of~$\prog$\longversion{~\cite{BaralGelfond94}}.}%
\longversion{%
} 
\longversion{%
}%
\longversion{%
}%
\begin{example}%
\label{ex:running1}\label{ex:running}
Consider %
$\prog\eqdef 
\SB 
\overbrace{ a \lor b \hsep}^{r_1};\, %
\overbrace{c \lor e \hsep}^{r_2};\, %
\overbrace{d \lor e \hsep b}^{r_3};\, %
\overbrace{b \hsep e, \neg d}^{r_4};\, %
\overbrace{d \hsep \neg b}^{r_5} %
\SE.$ %
It is easy to see that $\prog$ is a head-cycle-free program.
The set~$A=\{b, c, d\}$ is an answer set of~$\prog$.
Consider the ordering $\sigma =\langle b, c, d\rangle$, from which we
can prove atom~$b$ by rule~$r_1$, atom~$c$ by rule~$r_2$, and atom~$d$
by rule~$r_3$.
Further answer sets are $B=\{a,c,d\}$, $C=\{b,e\}$,
and~$D=\{a,d,e\}$.
\end{example}%
\noindent\textbf{Counting Projected Answer Sets.} %
An instance is a pair~$(\prog,P)$, where $\prog$ is a program and $P\subseteq\at(\prog)$
is a set of \emph{projection atoms}. %
The \emph{projected answer
  sets count} of~$\prog$ with respect to~$P$ is the number of
subsets~$I\subseteq P$ such that $I\cup J$ is an answer set of~$\prog$
for some set~$J\subseteq \at(\prog)\setminus P$.
The \emph{counting projected answer sets problem (\PASP)} asks to
output the projected answer sets count of~$\prog$,~i.e.,
$\Card{ \SB I \cap P \SM I \in S \SE}$ 
where $S$ is the set of all answer sets of~$\prog$.
\FIX{Note that~\cASP is~\PASP, where~$P=\at(\prog)$, and that deciding~\ASP equals~\PASP, where~$P=\emptyset$.}

\begin{example}\label{ex:running0}
  Consider program~$\prog$ from Example~\ref{ex:running} and its four
  answer sets $\{a,c,d\}$, $\{b,c,d\}$, $\{b,e\}$, and~$\{a,d,e\}$, as
  well as the set~$P\eqdef\{d, e\}$ of projection atoms.
  When we project the answer sets to~$P$, we only have the
  three answer sets $\{d\}$, $\{e\}$, and~$\{d, e\}$, i.e.,
  the projected answer sets
  count of~$(\prog,P)$ is 3.
\end{example}
\begin{theorem}[$\star$\footnote{Proofs marked with ``$\star$'' \longversion{are in the appendix}\shortversion{are in extended version at: \url{https://tinyurl.com/ycxe3zqo}}.}]\label{prop:hcfproj}
  The problem~$\PASP$ is \sharpSigma{2}-complete for
  disjunctive programs and \sharpNP-complete for %
  head-cycle-free, normal or tight programs.
\end{theorem}
\noindent\textbf{Tree Decompositions (TDs).} We follow standard terminology 
on graphs and digraphs\longversion{~\cite{Diestel12,BondyMurty08}}.  For a tree~$T=(N,A,n)$ with
root~$n$ and a node~$t \in N$, we let $\children(t, T)$ be the
sequence of all nodes~$t'$ in arbitrarily but fixed order, which have
an edge~$(t,t') \in A$.
Let $G=(V,E)$ be a graph.
A \emph{tree decomposition (TD)} of graph~$G$ is a pair
$\TTT=(T,\chi)$, where $T=(N,A,n)$ is a rooted tree, $n\in N$ the root,
and $\chi$ a mapping that assigns to each node $t\in N$ a set
$\chi(t)\subseteq V$, called a \emph{bag}, such that the following
conditions hold:
(i) $V=\bigcup_{t\in N}\chi(t)$ and
$E \subseteq\bigcup_{t\in N}\SB \{u,v\} \SM u,v\in \chi(t)\SE$; and (ii)
for each $r, s, t$, such that $s$ lies on the path from $r$ to
$t$, we have $\chi(r) \cap \chi(t) \subseteq \chi(s)$.
Then, $\width(\TTT) \eqdef \max_{t\in N}\Card{\chi(t)}-1$.  The
\emph{treewidth} $\tw{G}$ of $G$ is the minimum $\width({\TTT})$ over
all TDs $\TTT$ of $G$.
For arbitrary but fixed $w \geq 1$, it is feasible in linear time to
decide if a graph has treewidth at most~$w$ and, if so, to compute a
TD of width $w$~\cite{Bodlaender96}.
For simplifications \emph{we always use so-called nice TDs}, which can be
computed in linear time without increasing the width~\cite{\shortversion{BodlaenderKloks96}\longversion{Kloks94a}}
and are defined as follows.
For a node~$t \in N$, we say that $\type(t)$ is $\leaf$ if
$\children(t,T)=\langle \rangle$; $\join$ if
$\children(t,T) = \langle t',t''\rangle$ where
$\chi(t) = \chi(t') = \chi(t'') \neq \emptyset$; $\intr$
(``introduce'') if $\children(t,T) = \langle t'\rangle$,
$\chi(t') \subseteq \chi(t)$ and $|\chi(t)| = |\chi(t')| + 1$; $\rem$
(``removal'') if $\children(t,T) = \langle t'\rangle$,
$\chi(t') \supseteq \chi(t)$ and $|\chi(t')| = |\chi(t)| + 1$. \FIX{If for
every node $t\in N$, $\type(t) \in \{ \leaf, \join, \intr, \rem\}$, and~$\chi(t')=\emptyset$ for root and leaf~$t'$,
the TD is %
\emph{nice}.}

\begin{example}
  Figure~\ref{fig:graph-td} illustrates a graph~$G_1$ and a tree
  decomposition of~$G_1$ of width~$2$. By a
  property\longversion{\footnote{The vertices $e$,$b$,$d$ that are all neighbors to
    each other in~$G_1$.}} of tree decompositions~\cite{\shortversion{BodlaenderKloks96}\longversion{Kloks94a}}, the
  treewidth of~$G_1$ is~$2$.
\end{example}

\vspace{-1.3em}
\section{Dynamic Programming on TDs}
\vspace{-.45em}

In order to use TDs for ASP solving, we need a
dedicated graph representation of ASP
programs~\cite{FichteEtAl17a}.%
The \emph{primal graph}~$G_\prog$
of program~$\prog$ has the atoms of~$\prog$ as vertices and an
edge~$\{a,b\}$ if there exists a rule~$r \in \prog$ and $a,b \in \at(r)$.
\longerversion{The \emph{incidence graph}~$I_\prog$ of $\prog$ is the bipartite
graph that has the atoms and rules of~$\prog$ as vertices and an
edge~$a\, r$ if $a \in \at(r)$ for some rule~$r \in \prog$.}

\begin{figure}[t]%
  \centering
    \begin{tikzpicture}[node distance=7mm,every node/.style={fill,circle,inner sep=2pt}]
\node (a) [label={[text height=1.5ex,yshift=0.0cm,xshift=0.05cm]left:$e$}] {};
\node (d) [above right of=a,label={[text height=1.5ex,yshift=0.09cm,xshift=-0.05cm]right:$b$}] {};
\node (b) [left of=d,yshift=0.0em,xshift=0.6em,label={[text height=1.5ex,yshift=-0.25em]above:$a$}] {};
\node (c) [left of=d,xshift=-1em,label={[text height=1.5ex,yshift=0.09cm,xshift=0.05cm]left:$c$}] {};
\node (e) [right of=a, yshift=-0em,xshift=.75em,label={[text height=1.5ex]right:$d$}] {};
\draw (a) to (c);
\draw (b) to (d);
\draw (d) to (e);
\draw (e) to (a);
\draw (d) to (a);
\end{tikzpicture}%
    \includegraphics{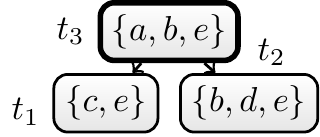}%
    \caption{Graph~$G_1$ and a tree decomposition of~$G_1$.}
  \longerversion{%
    \begin{subfigure}[c]{0.47\textwidth}
      \centering%
      \begin{tikzpicture}[node distance=7mm,every node/.style={fill,circle,inner sep=2pt}]
\node (a) [label={[text height=1.5ex,yshift=0.0cm,xshift=0.05cm]left:$e$}] {};
\node (d) [above right of=a,label={[text height=1.5ex,yshift=0.09cm,xshift=-0.05cm]right:$b$}] {};
\node (b) [left of=d,yshift=0.0em,xshift=0.6em,label={[text height=1.5ex,yshift=-0.25em]above:$a$}] {};
\node (c) [left of=d,xshift=-1em,label={[text height=1.5ex,yshift=0.09cm,xshift=0.05cm]left:$c$}] {};
\node (e) [right of=a, yshift=-0em,xshift=.75em,label={[text height=1.5ex]right:$d$}] {};
\draw (a) to (c);
\draw (b) to (d);
\draw (d) to (e);
\draw (e) to (a);
\draw (d) to (a);
\end{tikzpicture}%
      \input{graph0/td}%
      \caption{Graph~$G_1$ and a tree decomposition of~$G_1$.}
      \label{fig:graph-td}
    \end{subfigure}
    \begin{subfigure}[c]{0.5\textwidth}
      \centering \input{graph0/graph_inc}%
      \input{graph0/td_inc}%
      \caption{Graph~$G_2$ and a tree decomposition of~$G_2$.}
      \label{fig:graph-td2}%
    \end{subfqigure}
    \caption{Graphs~$G_1, G_2$ and two corresponding tree
      decompositions.}
  }%
  \label{fig:graph-td}%
\end{figure}
\begin{example}
  Recall program~$\prog$ from Example~\ref{ex:running1} and observe
  that graph~$G_1$ in Figure~\ref{fig:graph-td} is the primal graph~$G_\prog$
  of~$\prog$.
  \longerversion{%
    Further, graph~$G_2$ of Figure~\ref{fig:graph-td2} is the
    incidence graph of~$\prog$.
  }
\end{example}

\noindent Let ${\cal T} = (T, \chi)$ be a TD of primal
graph~$G_\prog$ of a program $\prog$. Further, let $T = (N,\cdot,n)$
and $t \in N$.  The \emph{bag-program} is defined as
$\prog_t \eqdef \SB r \SM r \in \prog, \at(r) \subseteq \chi(t)\}$, the \emph{program below $t$}
as
$\progt{t} \eqdef \SB r \SM r \in \prog_{t'}, t' \in \post(T,t) \SE$,
and the \emph{program strictly below $t$} as
$\progtneq{t}\eqdef \progt{t}\setminus \prog_t$. It holds that
$\progt{n} = \progtneq{n} = \prog$~\cite{FichteEtAl17a}. Analogously, we define the 
\emph{atoms below~$t$} by~$\att{t} \eqdef \cup_{t'\in\post(T,t)}\chi(t')$, and the \emph{atoms strictly below~$t$} by~$\attneq{t}\eqdef \att{t}\setminus\chi(t)$.
\longversion{For an
example we refer to Example~\ref{ex:bagprog}$^\star$.}

Algorithms that decide consistency or solve
\cASP\shortversion{~\cite{FichteEtAl17a}}\longversion{~\cite{FichteEtAl17a,JaklPichlerWoltran09}} proceed by \emph{dynamic
programming (DP)} along the TD (in post-order) where at each
node of the tree information is gathered~\cite{BodlaenderKloks96} in a
table by a \FIX{(local) \emph{table algorithm}~$\AlgA$.} %
More generally, a \emph{table} is a set of rows, where a
\emph{row}~$\vec\tabval$ is a sequence of fixed length. 
Similar as for sequences when addressing the $i$-th element, for a
set~$U$ of rows (table) we let
$U_{(i)}\eqdef\{\vec u_{(i)} \mid \vec u \in U\}$.
The actual length, content, and meaning of the rows depend on the
algorithm~$\AlgA$.  
Since we later traverse the TD
repeatedly running different algorithms, we explicitly state
\emph{$\AlgA$-row} if rows of this \emph{type} are syntactically used
for algorithm~$\AlgA$ and similar %
\emph{$\AlgA$-table} for tables.
In order to access tables computed at certain nodes after a traversal
as well as to provide better readability, we attribute TDs with an additional mapping to store tables.
Formally, a \emph{tabled tree decomposition (TTD)} of graph~$G$ is a
\FIX{triple} $\TTT=(T,\chi,\tau)$, where $(T,\chi)$ is a TD
of~$G$ and $\tau$ maps nodes~$t$ of~$T$ to tables. If not specified
otherwise, we assume that $\tau(t)= \{\}$ for every node~$t$ of~$T$.
\FIX{When a TTD has been computed using algorithm~$\AlgA$ after traversing
the TD, we call the decomposition the $\AlgA$-TTD of
the given instance.}
DP for ASP
performs the following steps: \\[0.21em]
\noindent 1. \FIX{Given program~$\prog$, compute a tree decomposition of the primal graph~$P_\prog$.}\\
\noindent 2. Run algorithm~$\dpa_\AlgA$ (see
  Listing~\ref{fig:dpontd}). 
  It takes a TTD~$\mathcal{T}=(T,\chi, \iota)$
  with~$T=(N,\cdot,n)$ and traverses~$T$ in post-order\footnote{$\text{post-order}(T,n)$ provides the sequence of nodes for tree~$T$ rooted at~$n$.}.
  At each node~$t \in N$ it computes a new $\AlgA$-table~$o(t)$ by
  executing the algorithm~$\AlgA$. Algorithm~$\AlgA$ has a ``local
  view'' on the computation and can access only $t$, the atoms in the
  bag~$\chi(t)$, the bag-program~$\prog_t$, and $\AlgA$-table~$o(t')$ for any child~$t'$ of $t$.\footnote{Note
    that in Listing~\ref{fig:dpontd}, $\AlgA$ takes in addition as
    input set~$P$ and table~$\iota_t$, used later. %
    Later, $P$ represents the projected atoms and
    $\iota_t$ is a table at~$t$ from an earlier traversal. }
    \FIX{Finally, $\dpa_\AlgA$ returns an $\AlgA$-TTD~$(T, \chi, o)$.}\\
\noindent 3. Print the result by interpreting table~$o(n)$ for root~$n$
  of~$T$.
\begin{algorithm}[t]%
  \KwData{%
    Problem instance~$(\prog,P)$, TTD~$\TTT=(T,\chi, \iota)$
    of~$G_\prog$ such that~$n$ is the root of~$T$,
    $\children(t, T) = \langle t_1, \ldots, t_\ell\rangle$. \textbf{Out: }$\AlgA$-TTD~$(T,\chi, o)$, $\AlgA$-table mapping $o$.\hspace{-1em}%
  }%

  \hspace{-0.1em}$o \leftarrow \text{empty mapping}$

  \For{\text{\normalfont iterate} $t$ in \text{\normalfont post-order}(T,n)}{
    \vspace{-0.05em}%

    \hspace{-0.5em}$o(t) \leftarrow {\AlgA}(t,
    \chi(t),\iota(t),(\prog_t,P),\langle o(t_1), \ldots,
    o(t_\ell) \rangle)$\hspace{-0.8em} %
    \vspace{-0.5em} %
  }%
  \Return{$(T, \chi, o)$}
  \captionsetup{format=hang}
  \caption{%
    Algorithm ${\dpa}_{\AlgA}((\prog,P), \TTT)$:
    Dynamic programming on TTD
    $\mathcal{T}$,~c.f.,~\protect\cite{FichteEtAl17a}.
  }%
\label{fig:dpontd}
\end{algorithm}%

Then, the actual computation of algorithm~$\AlgA$ is a somewhat
technical case distinction of the types~$\type(t)$ \FIX{we see} when
considering node~$t$.
Algorithms for counting answer sets of disjunctive
programs\longversion{~\cite{JaklPichlerWoltran09}} and its
extensions~\cite{FichteEtAl17a} have already been established.
Implementations of these algorithms can be useful also for 
solving~\cite{FichteEtAl17a\longversion{,FichteEtAl17b}}, but the running time is
clearly double exponential time in the treewidth in the worst case.
We, however, establish an algorithm ($\PRIM$) that is restricted to
head-cycle-free programs. The runtime of our algorithm is factorial in
the treewidth and therefore faster than previous algorithms. Our
constructions are inspired by ideas used in previous DP algorithms%
\shortversion{.}
\longversion{~\cite{PichlerEtAl14}.}
In the following, we first present the table algorithm for deciding whether a head-cycle-free program has an answer
set (\ASP).
In the end, this algorithm outputs a new TTD,
which we later reuse to solve the actual counting problem.  Note that
the TD itself remains the same, but for readability,
we keep computed tables and nodes aligned.

\longversion{\medskip}\noindent\textbf{Consistency of Head-Cycle-Free Programs.}
We can use algorithm~$\dpa_\PRIM$ to decide the consistency
problem~\ASP for head-cycle-free programs and simply specify our new table
algorithm (\PRIM) that ``transforms'' tables from one node to another.
As graph representation we use the primal graph.  The idea is to
implicitly apply along the TD the characterization of
answer sets by Lin and Zhao~\cite{LinZhao03} extended to head-cycle-free
programs~\cite{Ben-EliyahuDechter94}.
To this end, we store in table~$o(t)$ at each node~$t$ rows of the
form~$\langle I, \mathcal{P}, \sigma\rangle$.
The first position consists of an interpretation~$I$ restricted to the
bag~$\chi(t)$.  For a sequence~$\vec \tabval$, we
write~$\mathcal{I}(\vec \tabval)\eqdef \vec u_{(1)}$ to address the
\emph{interpretation part}.
The second position consists of a set~$\mathcal{P} \subseteq I$ that
represents atoms in~$I$ for which we know that they have already been
proven.
The third position~$\sigma$ is a sequence of the atoms \FIX{in}~$I$ such that
there is a super-sequence~$\sigma'$ of~$\sigma$, which induces an ordering~$<_{\sigma'}$.
Our table algorithm~\PRIM stores interpretation parts always
restricted to bag~$\chi(t)$ and ensures that an interpretation can be extended
to a model of sub-program~$\prog_{\leq t}$.
More precisely, it guarantees that interpretation~$I$ can be extended
to a model~$I'\supseteq I$ of~$\prog_{\leq t}$ and that the atoms
in~$I'\setminus I$ (and the atoms in $\mathcal{P}\subseteq I$) have
already been \emph{proven}, using some induced ordering~$<_{\sigma'}$
where $\sigma$ is a sub-sequence of~$\sigma'$.
In the end, an interpretation~$\mathcal{I}(\vec u)$ of a row~$\vec u$
of the table~$o(n)$ at the root~$n$ proves that there is a
superset~$I' \supseteq \mathcal{I}(\vec u)$ that is an answer set
of~$\prog = \progt{n}$.

Listing~\ref{fig:prim} presents the algorithm~\PRIM.  Intuitively,
whenever an atom~$a$ is introduced ($\intr$), we decide whether we
include~$a$ in the interpretation, determine bag atoms that can be
proven in consequence of this decision, and update the
sequence~$\sigma$ accordingly.
To this end, we define 
for a given interpretation~$I$ and a sequence~$\sigma$ the set
$\gatherproof(I, \sigma, \prog_t) \eqdef \cup_{r\in \prog_t, a\in
  H_r}\SB a \mid B_r^+ \subseteq I, I \cap B^-_r = \emptyset, I \cap
(H_r\setminus \{a\}) = \emptyset, B^+_r <_\sigma a \SE$ where
$B^+_r <_\sigma a$ holds if $b <_\sigma a$ is true for every
$b \in B^+_r$.  Moreover, given a
sequence~$\sigma=\langle \sigma_1, \ldots, \sigma_k\rangle$ and a
set~$A$ of atoms, we compute the potential sequences
involving~$A$. Therefore, we let
$\possord(\sigma, A) \eqdef \SB \sigma \SM A=\emptyset \SE \cup
\bigcup_{a\in A}\{\langle a, \sigma_1, \ldots, \sigma_k\rangle,
\ldots, \langle \sigma_1, \ldots, \sigma_k, a\rangle\SE$.
When removing ($\rem$) an atom~$a$, we only keep those rows where~$a$ has
been proven (contained in~$\mathcal{P}$) and then restrict remaining rows
to the bag (not containing~$a$). In case the node is of
type~$\join$, we combine two rows in two different child tables,
intuitively, we are enforced to agree on interpretations~$I$ and
sequences~$\sigma$. However, concerning individual
proofs~$\mathcal{P}$, it suffices that an atom is proven in \emph{one} of the rows.

\renewcommand{\eqdef}{\leftarrow}
%
%
 \begin{algorithm}[t]
   \KwData{Node~$t$, bag $\chi_t$, bag-program~$\prog_t$,
     $\langle \tab{1}, \ldots \rangle$ is the sequence of \PRIM-tables
     of~children of~$t$.{~\bf Out:} \PRIM-table~$\tab{t}.\shortversion{\hspace{-5em}}$}
   \lIf(\hspace{-1em})
   {$\type(t) = \leaf$}{%
     $\tab{t} \eqdef \{ \langle
     \tuplecolor{\inputPredColor}{\emptyset}, \tuplecolor{\outputPredColor}{\emptyset}, \tuplecolor{\statePredColor}{\langle\rangle}
     \rangle \}$\label{line:primleaf}%
     %
   }%
  \uElseIf{$\type(t) = \intr$ and $a\hspace{-0.1em}\in\hspace{-0.1em}\chi_t$ is the introduced atom}{
   \vspace{-0.05em}
   \makebox[0cm][l]{}\hspace{-1em}$\tab{t} \eqdef 
   \{ \langle \tuplecolor{\inputPredColor}{J}, \tuplecolor{\outputPredColor}{\mathcal{P} \cup \gatherproof(J, \sigma', \prog_t)}, \tuplecolor{\statePredColor}{\sigma'} \rangle$\newline
     \makebox[2.62cm][l]{}$|\;\langle \tuplecolor{\inputPredColor}{I}, \tuplecolor{\outputPredColor}{\mathcal{P}}, \tuplecolor{\statePredColor}{\sigma} \rangle\in \tab{1},$ $J \in \{I, \MAI{I}{a}\},$ 
     $J \models \prog_t, $
      $ \sigma' \in \possord(\sigma, \{a\} \cap J)\}
      \hspace{-5em}$\label{line:primintr}
   \vspace{-0.05em}
     }\vspace{-0.05em}%
     \uElseIf{$\type(t) = \rem$ and $a \not\in \chi_t$ is the removed atom}{%
       \makebox[2.8cm][l]{\hspace{-1em}$\tab{t} \eqdef 
       \{ \langle \tuplecolor{\inputPredColor}{\MAR{I}{a}}, \tuplecolor{\outputPredColor}{\MAR{\mathcal{P}}{a}}, \tuplecolor{\statePredColor}{\MARR{\sigma}{a}}
       \rangle$}$|\;\langle \tuplecolor{\inputPredColor}{I}, \tuplecolor{\outputPredColor}{\mathcal{P}}, \tuplecolor{\statePredColor}{\sigma}
       \rangle \in \tab{1}, a \in \mathcal{P} \cup (\{a\} \setminus I) \}
       \hspace{-5em}$\label{line:primrem}
       \vspace{-0.1em}
     } %
     \uElseIf{$\type(t) = \join$}{
       \makebox[2.8cm][l]{\hspace{-1em}$\tab{t} \eqdef 
       \{ \langle \tuplecolor{\inputPredColor}{I}, \tuplecolor{\outputPredColor}{\mathcal{P}_1 \cup \mathcal{P}_2}, \tuplecolor{\statePredColor}{\sigma}
         \rangle$}$|\;\langle \tuplecolor{\inputPredColor}{I}, \tuplecolor{\outputPredColor}{\mathcal{P}_1}, \tuplecolor{\statePredColor}{\sigma} \rangle \in \tab{1}, \langle \tuplecolor{\inputPredColor}{I}, \tuplecolor{\outputPredColor}{\mathcal{P}_2}, \tuplecolor{\statePredColor}{\sigma} \rangle \in \tab{2}\}\hspace{-5em}$\label{line:primjoin}
       \vspace{-0.1em}
     } 
     \Return $\tab{t}$
     \vspace{-0.25em}
     \caption{Table algorithm~$\PRIM(t, \chi_t, \cdot, (\prog_t, \cdot),
       \langle \tau_1, \ldots \rangle)$.}
 \label{fig:prim}\algorithmfootnote{
\renewcommand{\eqdef}{{\ensuremath{\,\mathrel{\mathop:}=}}}
  $\MARR{\sigma}{\sigma_{i}} \eqdef \langle \sigma_1, \ldots, \sigma_{i-1}, \sigma_{i+1}, \ldots, \sigma_{k}\rangle$ where~$\sigma=\langle\sigma_1, \ldots, \sigma_k\rangle$,
  $\MAI{S}{e} \eqdef S \cup \{e\}$, 
  $\MAR{S}{e} \eqdef S \setminus \{e\}$.}
\end{algorithm}%
\renewcommand{\eqdef}{{\ensuremath{\,\mathrel{\mathop:}=}}}
%
%
%
\shortversion{\vspace{-.25em}}
\begin{example}\label{ex:sat}%
  Recall program~$\prog$ from
  Example~\ref{ex:running}. Figure~\ref{fig:running2} depicts a
  TD~$\TTT=(T,\chi)$ of the primal graph~$G_1$ of $\prog$. Further,
  the figure illustrates a snippet of tables of the
  TTD~$(T,\chi,\tau)$, which we obtain when running $\dpa_{\PRIM}$ on
  program~$\prog$ and TD~$\TTT$ according to Listing~\ref{fig:prim}.
  In the following, we briefly discuss some selected rows of those
  tables.
  Note that for simplicity and space reasons, we write $\tau_j$
  instead of $\tau(t_j)$ and identify rows by their node and
  identifier~$i$ in the figure. For example, the row
  $\vec\tabval_{13.3}=\langle I_{13.3}, \mathcal{P}_{13.3},
  \sigma_{13.3}\rangle\in\tab{13}$ refers to the third row of
  table~$\tab{13}$ for node~$t_{13}$. 
  Node~$t_1$ is of type~$\leaf$. Table~$\tab{1}$ has only one row,
  which consists of the empty interpretation, empty set of proven
  atoms, and the empty sequence (Line~\ref{line:primleaf}).
  Node~$t_2$ is of type~$\intr$ and introduces atom~$a$. Executing
  Line~\ref{line:primintr} results in
  $\tab{2}=\{\langle \emptyset, \emptyset,\langle \rangle \rangle,
  \langle \{a\}, \emptyset, \langle a\rangle\rangle\}$.
  Node~$t_3$ is of type~$\intr$ and introduces~$b$. Then, bag-program
  at node~$t_3$ is $\prog_{t_3}=\{a \vee b \hsep\}$.
  By construction (Line~\ref{line:primintr}) we ensure that
  interpretation~$I_{3.i}$ is a model of~$\prog_{t_3}$ for every
  row~$\langle I_{3.i}, \mathcal{P}_{3.i}, \sigma_{3.i}\rangle$
  in~$\tab{3}$.
  Node~$t_{4}$ is of type~$\rem$.  Here, we restrict the rows such
  that they contain only atoms occurring in bag~$\chi(t_4)=\{b\}$.  To
  this end, Line~\ref{line:primrem} takes only
  rows~$\vec\tabval_{3.i}$ of table~$\tab{3}$ where atoms in~$I_{3.i}$
  are also proven,~i.e., contained in~$\mathcal{P}_{3.i}$.
  In particular, every row in table~$\tab{4}$ originates from at least
  one row in~$\tab{3}$ that either proves~$a\in \mathcal{P}_{3.i}$ or
  where~$a\not\in I_{3.i}$. 
  Basic conditions of a TD ensure that once an atom is removed, it
  will not occur in any bag at an ancestor node. Hence, we also
  encountered all rules where atom~$a$ occurs.
  Nodes~$t_5, t_6, t_7$, and~$t_8$ are symmetric to
  nodes~$t_1, t_2, t_3$, and~$t_4$.
  Nodes~$t_9$ and~$t_{10}$ again introduce atoms. 
  Observe that $\mathcal{P}_{10.4} = \{e\}$ since
  $\sigma_{10.4}$ does not allow to prove~$b$ using atom~$e$.
  However, $\mathcal{P}_{10.5}=\{b,e\}$ as the sequence~$\sigma_{10.5}$
  allows to prove~$b$.
  In particular, in row~$\vec\tabval_{10.5}$ atom~$e$ is used to
  derive~$b$.  As a result, atom~$b$ can be proven, whereas
  ordering~$\sigma_{10.4}=\langle b,e\rangle$ does not serve in
  proving~$b$.
  We proceed similar for nodes~$t_{11}$ and $t_{12}$.
  At node~$t_{13}$ we join tables~$\tab{4}$ and~$\tab{12}$ according
  to Line~\ref{line:primjoin}.
  \FIX{Finally, $\tab{14}\neq \emptyset$, i.e., $\prog$ has an
  answer set; joining
  interpretations~$I$ of yellow marked rows of
  Figure~\ref{fig:running2} leads to~$\{b,e\}$.}
\end{example}
\setlength{\tabcolsep}{0.25pt}
\renewcommand{\arraystretch}{0.75}
\begin{figure}[t]
\centering%
\hspace{-0.75em}
\includegraphics[width=0.913\textwidth]{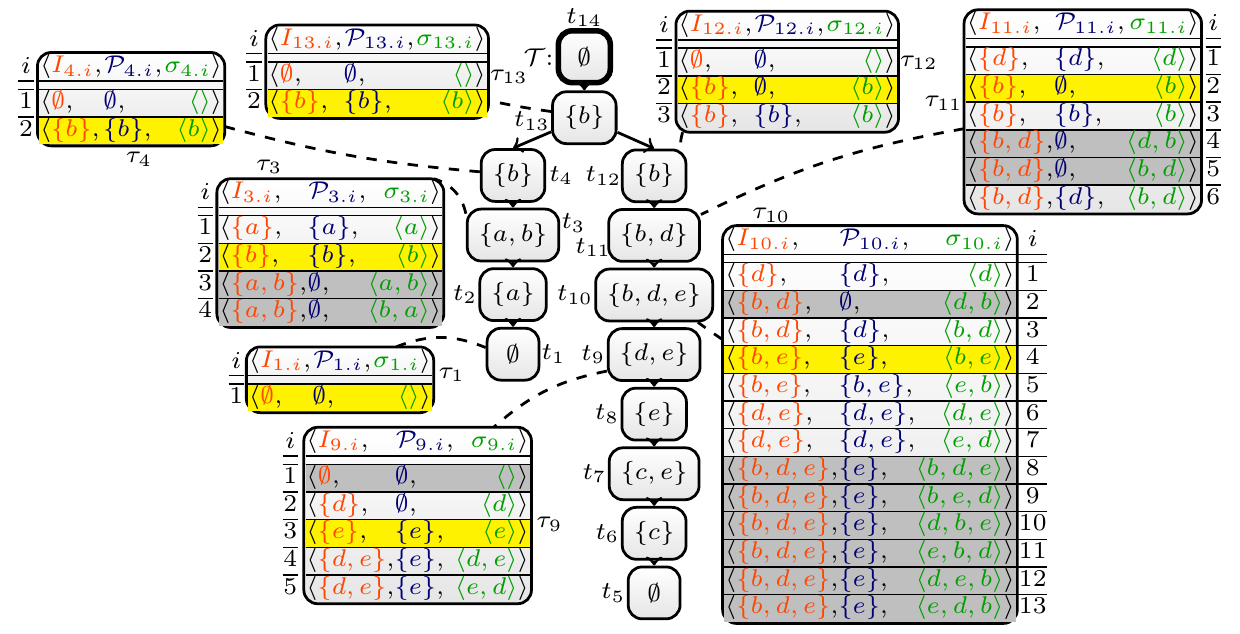}%
\caption{Selected tables of~$\tau$ obtained by~$\dpa_{\PRIM}$ on
  TD~${\cal T}$.} %
\label{fig:running2}
\end{figure}

\newcommand{\llangle}{\ensuremath{\langle\hspace{-2pt}\{\hspace{-0.2pt}}}%
\newcommand{\rrangle}{\ensuremath{\}\hspace{-2pt}\rangle}}
\newcommand{\STab}{\ensuremath{\ATab{\AlgA}}}%
\shortversion{\vspace{-.25em}}
\noindent Next, we provide a notion to reconstruct answer sets from a computed
TTD, which allows for computing for a given row its predecessor rows
in the corresponding child tables,~c.f., \cite{FichteEtAl18}.
Let $\prog$ be a program, $\TTT=(T, \chi, \tau)$ be an~$\AlgA$-TTD
of~$G_\prog$, and $t$ be a node of~$T$ where
$\children(t,T)=\langle t_1, \ldots, t_{\ell}\rangle$. %
Given a sequence $\vec s=\langle s_1, \ldots, s_{\ell} \rangle$, we
let
$\llangle \vec s\rrangle \eqdef \langle \{s_1\}, \ldots, \{s_{\ell}\}
\rangle$. %
For a given $\AlgA$-row~$\vec u$, we define the originating
$\AlgA$-rows of~$\vec u$ in node~$t$ by
$\orig(t,\vec \tabval) \eqdef \SB \vec s \SM \vec s \in \tau(t_1)
\times \cdots \times \tau({t_\ell}), \vec u \in {\AlgA}(t,\chi(t),
\cdot,(\prog_t,\cdot), \llangle \vec s\rrangle) \SE.$ %
We extend this to an $\AlgA$-table~$\rho$ by
$\origs(t,\rho) \eqdef \bigcup_{\vec u \in \rho}\orig(t,\vec u)$.
\shortversion{\vspace{-.25em}}
\begin{example}\label{ex:origins} %
  Consider program~$\prog$ and $\PRIM$-TTD~$(T,\chi,\tab{})$ from Example~\ref{ex:sat}.  We focus
  on~$\vec{\tabval_{1.1}} =\langle\emptyset, \emptyset, \langle
  \rangle\rangle$ of table~$\tab{1}$ of leaf~$t_1$. The
  row~$\vec{\tabval_{1.1}}$ has no preceding row,
  since~$\type(t_1)=\leaf$. Hence, we have
  $\origse{\PRIM}(t_1,\vec{\tabval_{1.1}})=\{\langle \rangle\}$.
  The origins of row~$\vec{\tabval_{11.1}}$ of
  table~$\tab{11}$ are given by
  $\origse{\PRIM}(t_{11},\vec{\tabval_{11.1}})$, which correspond to
  the preceding rows in
  table~$\tab{10}$ that lead to
  row~$\vec{\tabval_{11.1}}$ of
  table~$\tab{11}$ when running
  algorithm~$\PRIM$, i.e.,
  $\origse{\PRIM}(t_{11},\vec{\tabval_{11.1}}) = \{\langle
  \vec{\tabval_{10.1}} \rangle,\langle \vec{\tabval_{10.6}} \rangle,
  \langle \vec{\tabval_{10.7}}
  \rangle\}$. Origins of
  row~$\vec{{\tabval}_{12.2}}$ are given by
  $\origse{\PRIM}(t_{12},\vec{\tabval_{12.2}}) = \{\langle
  \vec{\tabval_{11.2}} \rangle, \langle
  \vec{\tabval_{11.6}} \rangle\}$. Note
  that~$\vec{\tabval_{11.4}}$
  and~$\vec{\tabval_{11.5}}$ are not among those origins, since
  $d$ is not proven.  Observe that
  $\origse{\PRIM}(t_j,\vec\tabval)=\emptyset$ for any
  row~$\vec\tabval\not\in\tab{j}$.
  For node~$t_{13}$ of type~$\join$ and row~$\vec{\tabval_{13.2}}$, 
  $\origse{\PRIM}(t_{13},\vec{\tabval_{13.2}})$ $=
  \{\langle\vec{\tabval_{4.2}},$ $\vec{\tabval_{12.2}} \rangle, \langle\vec{\tabval_{4.2}},$ $\vec{\tabval_{12.3}} \rangle\}$.
\end{example}

\longerversion{\paragraph{Table Algorithm for the Incidence Graph}

\input{algorithms/sinc}%

With the general algorithm in mind (see Figure~\ref{fig:dp-approach}), we are now ready to propose $\INC$, a
new table algorithm for solving \ASP on the semi-incidence graph (see 
Listing~\ref{fig:sinc}). 
As in the general approach, \INC computes and stores witness sets, and their
corresponding counter-witness sets. However, in addition, for each witness set
and counter-witness set, respectively, 
we need to store so-called \emph{satisfiability states}
(or \emph{sat-states}, for short), since the atoms of a rule may no longer be
contained in one single bag of the tree decomposition of the semi-incidence graph. Therefore, we
need to remember in each tree decomposition node, ``how much'' of a rule is already satisfied. The
following describes this in more detail.

By definition of tree decompositions and the semi-incidence graph, for
every atom~$a$ and every rule~$r$ of a program, it is true that if atom~$a$ occurs
in rule~$r$, then $a$ and $r$ occur together in at least one bag of
the tree decomposition. As a consequence, the table algorithm encounters every
occurrence of an atom in any rule. In the end, on removal of~$r$, we
have to ensure that $r$ is among the rules that are already
satisfied. However, we need to keep track of whether a witness satisfies a rule, because not all atoms that occur in a rule
occur together in a bag. Hence, when our algorithm traverses
the tree decomposition and an atom is removed we still need to store this
sat-state, as setting the removed atom to a certain truth
value influences the satisfiability of the rule.
Since the semi-incidence graph contains a clique on every set~$A$ of
atoms that occur together in a weight rule body or
choice rule head, %
those atoms~$A$ occur together in a bag in every tree decomposition of the
semi-incidence graph. For that reason, we do {not} need to
incorporate weight or 
choice rules into the satisfiability state, in
contrast to the table algorithm for the incidence graph discussed later~(c.f. Section~\ref{sec:inc}).

In algorithm~\INC (detailed in Listing~\ref{fig:sinc}), a tuple in the
table~$\tab{t}$ is a triple~$\langle M, \sigma, \CCC \rangle$.  The
set~$M \subseteq \at(\prog)\cap\chi(t)$ represents a witness set. 
The family~$\CCC$ of sets represents counter-witnesses, which we will
discuss in more detail below.
The sat-state~$\sigma$ for $M$ represents rules of $\chi(t)$ satisfied
by a superset of~$M$.  Hence, $M$ witnesses a model~$M'\supseteq M$
where $M' \models \progtneq{t} \cup \sigma$.  
We use the binary operator~$\cup$ to combine sat-states, which ensures
that rules satisfied in at least one operand remain satisfied. For a node~$t$, our algorithm considers a local-program depending on the bag~$\chi(t)$. Intuitively, this provides a local view on the program.

For a node~$t$, our algorithm considers a local-program depending on the bag~$\chi(t)$. Intuitively, this provides a local view on the program.
\begin{definition}\label{def:bagprogram}%
  Let $\prog$ be a program, $\TTT=(\cdot,\chi)$ a tree decomposition of $S(\prog)$,
  $t$ a node of $\TTT$ and ${R} \subseteq \prog_t$.
  The \emph{local-program mapping}~${R}^{(t)}: R \rightarrow prog(\at(R)\cap\chi(t))$ assigns to each rule~$r\in R$ a rule obtained from~$r$
  by %
  removing all
    literals~$a$ and $\neg a$ where $a \not\in \chi(t)$.
  \shortversion{}%
\end{definition}%

\begin{example}
  Observe
  $\prog_{t_1}^{(t_1)} = \{(r_1, b \hsep \neg a), (r_2, a \hsep \neg b)\}$ and
  $\prog_{t_2}^{(t_2)} = \{(r_1, b \hsep \neg a), (r_2, a \hsep \neg b), (r_3, d\hsep)\}$ for $\prog_{t_1}$, $\prog_{t_2}$ of
  Figure~\ref{fig:graph-td2}. %
\end{example}

Since the local-program mapping~$\prog^{(t)}$ depends on the considered
node~$t$, we may have different local-program mappings for node~$t$ and
its child~$t'$. In particular, the mappings~$\{r\}^{(t)}$ and
$\{r\}^{(t')}$ might already differ for a
rule~$r \in \chi(t) \cap \chi(t')$. In consequence for satisfiability
with respect to sat-states, we need to keep track of a representative
of a rule. 

$\SP(\dot{R}, \sigma) \eqdef \{ a \mid (r, s) \in \dot{R}, a \in H_s, a >_\sigma r\}$

$\checkord(\dot{R}, \sigma, \phi, a) \eqdef \text{true iff } a >_\sigma r \implies r\not\in\phi \text{ for any } (r,s) \in\dot{R}$ where $a\in\at(s)$

$\checkmod(\dot{R}, J, \sigma) \eqdef \text{true iff } J \cap \at(s)  = B_s^+ \cup X$, $\Card{X}\leq 1$, and~$X\subseteq H_s$ where $X = \{a\mid a \in \at(s), a >_\sigma r\} \text{ for any } (r,s) \in\dot{R} \text{ with } \sigma = \langle \ldots, r, \ldots \rangle$

Note that in the end~$X \cap H_s\neq \emptyset$ in at least some bag,
since otherwise the rule is not satisfied anyway, but we also have to
enforce that no body atom can in~$X$!}

Next, we provide statements on correctness and a runtime analysis. %

\shortversion{\vspace{-0.4em}}\begin{theorem}[$\star$%
]\label{thm:primcorrectness}
  \longversion{The algorithm~$\dpa_\PRIM$ is correct. %
  In other words, given}\shortversion{Given} a head-cycle-free program~$\prog$ and a
  TTD~${\cal T} = (T,\chi,\cdot)$ of~$G_\prog$
  where~$T=(N,\cdot,n)$ with root~$n$. Then,
  $\dpa_\PRIM((\prog,\cdot),\TTT)$ returns the
  $\PRIM$-TTD~$(T,\chi,\tau)$ such that $\prog$ has an answer set if
  and only if
  $\langle \emptyset, \emptyset, \langle
  \rangle\rangle\in\tau(n)$. \longversion{Further, we can construct all the answer
  sets of~$\prog$ from transitively following the origins
  of~$\tau(n)$.}  %
\end{theorem}%
\longversion{\begin{proof}[Proof (Idea)]
  For soundness, we state and establish an invariant %
  for every node~$t\in N$.  
  \longversion{For each
  row~$\vec\tabval=\langle I, \mathcal{P}, \sigma\rangle\in\tau(t)$,
  we have~$I\subseteq\chi(t), \mathcal{P}\subseteq I$, and~$\sigma$ is
  a sequence over atoms in~$I$. 
  Intuitively, we ensure existence of~$I'\supseteq I$ s.t.\
  $I'\models\progt{t}$ and that exactly the atoms in~$\attneq{t}$
  and~$\mathcal{P}$ can be proven using a super-sequence~$\sigma'$
  of~$\sigma$.  By construction, we guarantee that we can decide
  consistency if
  row~$\langle \emptyset, \emptyset, \langle
  \rangle\rangle\in\tau(n)$. Further, we can even reconstruct answer
  sets, by following $\origa{\PRIM}$ of this single row back to the
  leaves.}
  For completeness, we show that we indeed obtain 
  all required rows\longversion{ to
  output all the answer sets of~$\prog$}.
\end{proof}}
\shortversion{\vspace{-.9em}\begin{theorem}[$\star$]}
\longversion{\begin{theorem}}
  \label{thm:primruntime}
  Given a head-cycle-free program~$\prog$ and a TD~${\cal T} = (T,\chi)$ of~$G_\prog$ of width~$k$ with $g$
  nodes. Algorithm~$\dpa_{\PRIM}$ runs in time
  $\mathcal{O} (3^{k}\cdot k !  \cdot g)=\mathcal{O}(2^{k\cdot \text{log}(k)}\cdot g)$.
\end{theorem}
\longversion{
\begin{proof}[Proof (Sketch)]
  Let~$d = k+1$ be maximum bag size of the tree
  decomposition~$\TTT$. %
  The table~$\tau(t)$ has at most
  $3^{d} \cdot d!$ rows, since for a row~$\langle I, \mathcal{P}, \sigma\rangle$ we have~$d!$ many sequences~$\sigma$, and by construction of algorithm~$\PRIM$, an atom can be either in~$I$, both in~$I$ and~$\mathcal{P}$, or neither in~$I$ nor in~$\mathcal{P}$.
  In total, with the help of efficient data structures, e.g., for nodes~$t$ with~$\type(t)=\join$, one can establish a runtime bound of~$\bigO{{3^{d}\cdot d!}}$.
  Then, we apply this to every node~$t$ of the tree decomposition,
  which resulting in running
  time~$\bigO{{3^{d}\cdot d!} \cdot g}\subseteq \bigO{3^{k}\cdot k!\cdot g}$.
\end{proof}}
\longversion{
In order to obtain an upper bound on width factorial~$k!$, we can simply take
$k! \leq 2^k$ for any fixed width~$k\leq 3$. While in general~$k!$ is obviously not bounded by~$2^k$ for any fixed~$k\geq 4$, \emph{asymptotically}
$k!$ is bounded by~$2^{k^{(c+1)/c}}$ for any fixed positive integer~$c\geq 1$ as stated in Lemma~\ref{prop:kfact}.

\begin{lemma}[$\star$]\label{prop:kfact}
  Given any fixed positive integer~$c \geq 1$ and
  functions~$f(k)\eqdef k!, g(k) \eqdef 2^{k^{(c+1)/c}}$, where~$k$ is a non-negative integer. Then,
  $f \in \bigO{g}$.
\end{lemma}

In particular, $k!\leq 2^{k^{4/3}}$ for~$k \geq 342$, $k!\leq 2^{k^{5/4}}$ for~$k \geq 34556$, and $k!\leq 2^{k^{6/5}}$ for~$k \geq 3413636$.}

\shortversion{\vspace{-0.5em}}A natural question is whether we can significantly improve this
algorithm for fixed~$k$.  To this end, we take the \emph{exponential
  time hypothesis (ETH)} into account~\cite{ImpagliazzoPaturiZane01}, 
  which states that there is some real~$s > 0$ such that we cannot decide satisfiability of a given
3-CNF formula~$F$ in
time~$2^{s\cdot\Card{F}}\cdot\CCard{F}^{\mathcal{O}(1)}$.

\shortversion{\vspace{-0.2em}
\begin{proposition}[$\star$]}
\longversion{\begin{proposition}}
  Unless ETH fails, consistency of \FIX{head-cycle-free, normal or tight} program~$\prog$
  cannot be decided in time~$2^{o(k)} \cdot \CCard{\prog}^{o(k)}$
  where~$k=\tw{G_\prog}$. %
\end{proposition}
\longversion{%
\begin{proof}
  \FIX{Reduction from} SAT to \ASP similar to the proof of Theorem~\ref{prop:hcfproj}.
\end{proof}%
}
\noindent\FIX{In the construction above, we store an arbitrary but fixed ordering~$\sigma$ on
the involved atoms. We believe that we cannot avoid these orderings in
general, since we have to compensate arbitrarily ``bad'' orderings
induced by the decomposition. 
Hence, we claim that \ASP for head-cycle-free programs
is slightly superexponential, rendering our algorithm asymptotically worst-case optimal. Lokshtanov, Marx and Saurabh confirm
such an expectation~\cite{LokshtanovMarxSaurabh11}  whenever orderings are required.}
\begin{conjecture}
  Unless ETH fails, consistency of a head-cycle-free program~$\prog$
  cannot be decided in
  time~$2^{o(k\cdot \text{log}(k))} \cdot \CCard{\prog}^{o(k)}$
  where~$k=\tw{G_\prog}$. %
\end{conjecture}

\shortversion{\vspace{-1em}}
\section{Dynamic Programming for~$\PASP$}%

\label{sec:projmodelcounting}

\longversion{
\begin{figure}[t]
\centering
\includegraphics[scale=0.72]{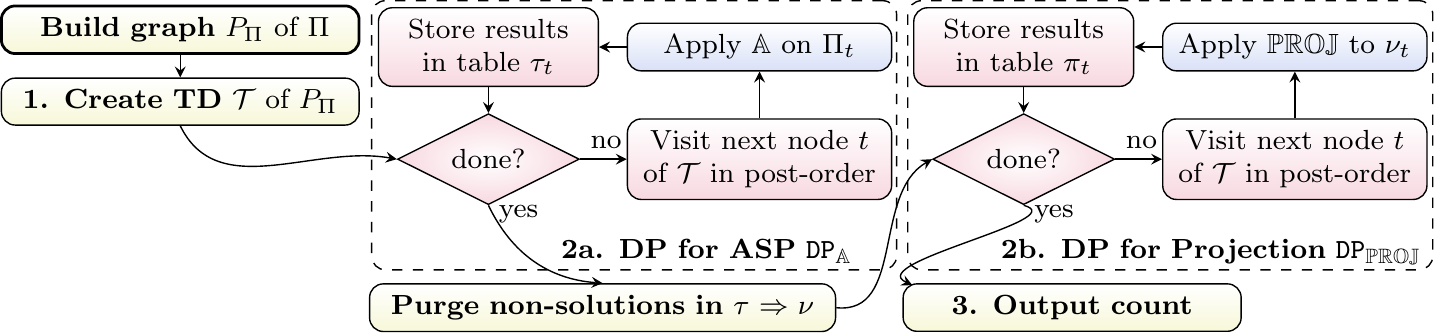}%
\caption{Algorithm~$\mdpa{\AlgA}$ consists of~$\dpa_\AlgA$
  and~$\dpa_\PROJ$. %
}
\label{fig:multiarch}
\end{figure}%
}

In this section, we present our DP
algorithm\footnote{\label{foot:phc}Later we use (among
  others)~\mdpa{\PRIM} where~$\AlgA=\PRIM$.}~\mdpa{\AlgA}, which
allows for solving the projected answer set counting problem (\PASP).
\mdpa{\AlgA} is based on an approach of projected counting for Boolean
formulas~\cite{FichteEtAl18} where TDs are traversed multiple times.
We show that ideas from that approach can be fruitfully extended to
answer set programming.
\longversion{Figure~\ref{fig:multiarch} illustrates the steps of \mdpa{\AlgA}.}
First, we construct the primal graph~$G_\prog$ of the input program~$\prog$
and compute a TD of $\prog$. Then, we traverse the TD a first time by
running $\dpa_\AlgA$ (Step~2a), which outputs a
TTD~$\TTT_{\text{cons}}=(T,\chi,\tau)$, where~$T=(N, \cdot, n)$.
Afterwards, we traverse $\TTT_{\text{cons}}$ in pre-order and remove
all rows from the tables %
that cannot be extended to an answer set (\emph{``Purge
  non-solutions''}).
In other words, we keep only rows~$\vec u$ of table~$\tau(t)$ at
node~$t$, if~$\vec u$ is involved in those rows that are used to
construct an answer set of~$\prog$, and let the resulting TTD\footnote{\label{foot:nu}Table~$\nu(t)$ contains rows
    obtained by recursively following origins of~$\tau(n)$ for root~$n$.\longversion{ Formal details are in Definition~\ref{def:extensions}$^\star$.}}
be~$\TTT_{\text{purged}}=(T,\chi,\nu)$. We refer to $\nu$
as~\emph{purged table mapping}.
In Step~2b ($\dpa_\PROJ$), we traverse $\TTT_{\text{purged}}$ to count
interpretations with respect to the projection atoms and obtain
$\TTT_{\text{proj}}=(T,\chi,\pi)$. From the table~$\pi(n)$ at the root~$n$ of
$T$, we can then read the projected answer sets count of the input instance.
In the following, we only describe the table algorithm $\PROJ$,
since the traversal in $\dpa_\PROJ$ is the same as before.
For \PROJ, %
a row at a node~$t$ is a pair $\langle\rho, c \rangle\in\pi(t)$, where
$\rho \subseteq \nu(t)$ is an $\AlgA$-table and $c$ is a non-negative
integer.
In fact, integer~$c$ stores the number of intersecting solutions
($\ipmc$). However, we aim for the projected answer sets count
($\pmc$), whose computation requires
\shortversion{%
  to extend previous definitions~\cite{FichteEtAl18}.
}
\longversion{%
a few additional
definitions. Therefore, we can simply widen definitions from very
recent work~\cite{FichteEtAl18}.
}
\noindent In the remainder, %
we assume~$(\prog, P)$ to be an instance of~\PASP, $(T, \chi, \tau)$
to be an $\AlgA$-TTD of~$G_\prog$ and the mappings~$\tau$, $\nu$, and
$\pi$ as used above. Further, let~$t$ be a node of~$T$ with~$\children(t,T)=\langle t_1, \ldots, t_\ell\rangle$ and let $\rho \subseteq \nu(t)$.
  The relation~$\bucket \subseteq \rho \times \rho$ considers
  equivalent rows with respect to the projection of its
  interpretations by %
  $\bucket \eqdef \SB (\vec u,\vec v) \SM \vec u, \vec v \in \rho,
  \restrict{\mathcal{I}(\vec u)}{P} = \restrict{\mathcal{I}(\vec
    v)}{P}\SE.$
  \FIX{Let $\buckets_P(\rho)$ be equivalence classes induced
  by~$\bucket$ on~$\rho$,~i.e.,
  $\buckets_P(\rho) \eqdef\, (\rho / \bucket) = \SB [\vec u]_P \SM
  \vec u \in \rho\SE$, where
  $[\vec u]_P = \SB \vec v \SM \vec v \bucket \vec u,\vec v \in
  \rho\}$\longversion{~\cite{Wilder12a}}.
  Further, 
  $\subbuckets_P(\rho) \eqdef \cup_{S \mid \emptyset \neq S \subseteq \buckets_P(\rho)}\{S\}$.}%
\begin{example}\label{ex:equiv} %
  Consider program~$\prog$, set~$P$,
  TTD~$(T,\chi, \tau)$, and table~$\tab{10}$ from
  Example~\ref{ex:running0} and Figure~\ref{fig:running2}.
  Rows~$\vec {u_{10.2}}$ and
  $\vec {u_{10.8}}, \ldots, \vec {u_{10.13}}$ are removed (highlighted gray) during purging,
  since they are not involved in any answer set, resulting in~$\nu_{10}$.
  Then, $\vec{ u_{10.4}} =_P \vec{ u_{10.5}}$ and
  $\vec{ u_{10.6}} =_P \vec{ u_{10.7}}$.  The set~$\nu_{10}/\bucket$ of equivalence classes 
  of $\nu_{10}$
  is~$\buckets_P(\nu_{10})=\SB \{\vec{ u_{10.1}}\}, \{\vec{ u_{10.3}}\}, \{\vec{ u_{10.4}}, \vec{ u_{10.5}}\}, \{\vec{
    u_{10.6}}, \vec{ u_{10.7}}\}\SE$.
\end{example}
%


  
\begin{algorithm}[t]
  \KwData{%
	Node~$t$, purged table mapping~$\nu_{t}$, projection atoms~$P$, sequence~$\langle \pi_1, \ldots \rangle$ of $\PROJ$-tables
     of~children of~$t$. \textbf{Out: }\PROJ-table~$\pi_{t}$ of pairs~$\langle \rho,
    c\rangle$, $\rho \subseteq \nu_{t}$, $c \in
    \Nat$.\hspace{-5em}
    %
  }%
    %
  %
  %
  $\makebox[0em]{}\pi_{t}\hspace{-0.2em}\leftarrow\hspace{-0.2em}\big\SB\langle \rho,
  \ipmc(t,\rho,\langle \pi_{1}, \ldots\rangle) \rangle \big{|}\, \rho \in
  \subbuckets_P(\nu_{t})\big\SE$
  \qquad\qquad\qquad\qquad\Return{$\pi_{t}$}\hspace{-5em}
  \vspace{-0.15em}
  \caption{Table algorithm $\PROJ(t, \cdot, \nu_t, (\cdot, P),
       \langle \pi_1, \ldots \rangle)$ for projected
    counting.}
  \label{fig:dpontd3}
\end{algorithm}


%
%

%
%
%
\shortversion{\vspace{-0.1em}}
Later, we require to construct already computed projected counts for
tables of children of a given node~$t$. Therefore, we define the
\emph{stored $\ipmc$} of a table~$\rho \subseteq \nu(t)$ in
table~$\pi(t)$ by
$\sipmc(\pi(t), \rho) \eqdef \Sigma_{\langle \rho, c\rangle \in \pi(t)}
c.$ 
We extend this to a
sequence~$s=\langle \pi(t_1), \ldots, \pi(t_\ell)\rangle$ of tables of
length $\ell$ and a
set~$O = \{\langle \rho_1, \ldots, \rho_\ell\rangle, \langle \rho_1',
\ldots, \rho_\ell'\rangle, \ldots\}$ of sequences of~$\ell$ tables by
$\sipmc(s, O)=\Pi_{i \in \{1, \ldots,
  \ell\}}\sipmc(s_{(i)},O_{(i)}).$
So we select the $i$-th position of the sequence together
with sets of the $i$-th positions. %

Intuitively, when we are at a node~$t$ in algorithm~$\dpa_\PROJ$ we
have already computed~$\pi(t')$ of $\TTT_{\text{proj}}$ for every node~$t'$
below~$t$.
Then, we compute the projected answer sets count
of~$\rho \subseteq \nu(t)$. Therefore, we apply the
inclusion-exclusion principle to the stored projected answer sets count
of origins.
We define $\pcnt(t,\rho, \langle\pi(t_1),\ldots\rangle) \eqdef %
\Sigma_{\emptyset \subsetneq O \subseteq {\origs(t,\rho)}}$
$(-1)^{(\Card{O} - 1)} \cdot \sipmc(\langle \pi(t_1), \ldots\rangle, O)$. %
Intuitively, %
$\pcnt$ determines the $\AlgA$-origins of table~$\rho$, 
goes over all subsets of these origins and looks up 
stored counts ($\sipmc$) in \PROJ-tables of children~$t_i$ of~$t$.

\begin{example}\label{ex:pcnt} %
  Consider again program~$\prog$ and TD~$\TTT$ from
  Example~\ref{ex:running1} and Figure~\ref{fig:running2}. First, we
  compute the projected count $\pcnt(t_{4},\{\vec{ u_{4.1}}\}, \langle\pi(t_{3})\rangle)$
  for row~$\vec{ u_{4.1}}$ of table~$\nu(t_{4})$, where
  $\pi(t_3) \eqdef\allowdisplaybreaks[4] \big\SB
  \langle \{\vec{ u_{3.1}}\}, 1\rangle,$
  $\langle \{\vec{ u_{3.2}}\},1\rangle, \langle \{\vec{u_{3.1}}, \vec{
    u_{3.2}}\},1\rangle\big\SE$ with
  $\vec{u_{3.1}}=\langle \emptyset, \emptyset, \langle\rangle \rangle$
  and~$\vec{u_{3.2}}=\langle \{a\}, \emptyset, \langle a\rangle
  \rangle$.
  Note that~$t_5$ has only the child~$t_4$ and therefore the product in~$\sipmc$
  consists of only one factor. 
  Since
  $\origse{\PRIM}(t_4, \vec{ u_{4.1}}) = \{\langle\vec{
    u_{3.1}}\rangle\}$, only the value of~$\sipmc$ for
  set~$\{\langle\vec{ u_{3.1}}\rangle\}$ is non-zero. Hence, we obtain
  $\pcnt(t_4,\{\vec{ u_{4.1}}\}, \langle\pi(t_3)\rangle)$ $=1$. 
  Next, we compute
  $\pcnt(t_{4},\{\vec{ u_{4.1}}, \vec{u_{4.2}}\}, \langle\pi(t_3)\rangle)$. Observe that
  $\origse{\PRIM}(t_4, \{\vec{ u_{4.1}}, \vec{ u_{4.2}}\}) =
  \{\langle\vec{ u_{3.1}}\rangle, \langle\vec{ u_{3.2}}\rangle\}$. We
  sum up the values of~$\sipmc$ for sets~$\{\vec{ u_{4.1}}\}$
  and~$\{\vec{ u_{4.2}}\}$ and subtract the one for
  set~$\{\vec{ u_{4.1}}, \vec{ u_{4.2}}\}$.  Hence, we obtain
  $\pcnt(t_4,\{\vec{ u_{4.1}}, \vec{ u_{4.2}}\}, \langle\pi(t_3)\rangle)=1+1-1=1$.
\end{example}

\shortversion{\vspace{-.5em}}
\noindent Next, we provide a definition to compute $\ipmc$ %
at a node~$t$ for given table~$\rho\subseteq \nu(t)$ by
computing %
$\pmc$ for children~$t_i$ of~$t$ using stored $\ipmc$
values from tables~$\pi(t_i)$, and subtracting and adding~$\ipmc$ values
for subsets~$\emptyset\subsetneq\varphi\subsetneq\rho$ accordingly.
Formally, $\icnt(t,\rho,s)\eqdef 1$ if $\type(t) = \leaf$ and
otherwise
$\icnt(t,\rho,s)\eqdef \big|\pcnt(t,\rho, s)$ $+
\Sigma_{\emptyset\subsetneq\varphi\subsetneq\rho}(-1)^{\Card{\varphi}}
\cdot \ipmc(t,\varphi, s)\big|$ where
$s = \langle \pi(t_1), \ldots\rangle$.
In other words, if a node is of type~$\leaf$ the $\ipmc$ is one, since
bags of leaf nodes are empty.
Otherwise, we compute the ``non-overlapping'' count of given
table~$\rho\subseteq\nu(t)$ with respect to~$P$, by exploiting %
inclusion-exclusion principle on $\AlgA$-origins of~$\rho$ such that
we count every projected answer set only once. Then we have to %
subtract and add $\ipmc$ values (``all-overlapping'' counts) for
strict subsets~$\varphi$ of~$\rho$, accordingly.
Finally, Listing~\ref{fig:dpontd3} presents table algorithm~\PROJ,
which stores~$\pi(t)$ consisting of every sub-bucket of %
given table~$\nu(t)$ together with its $\ipmc$.

\setlength{\tabcolsep}{0.25pt}
\renewcommand{\arraystretch}{0.91}
\begin{figure*}[t]
\hspace{-0.4em}%
\centering%
\hspace{-0.75em}
\includegraphics[width=0.913\textwidth]{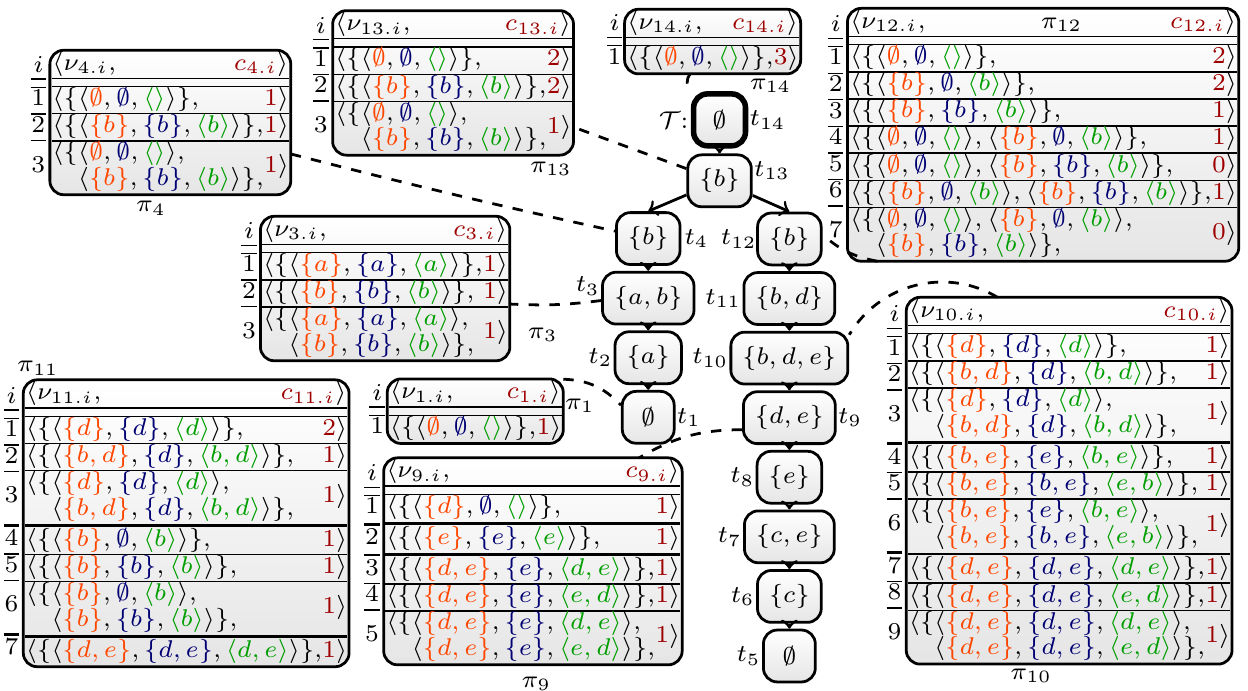}%
\caption{Selected tables of~$\pi$ obtained by~$\dpa_{\algo{PROJ}}$ on
  TD~${\cal T}$ and purged table mapping~$\nu$ (obtained by purging on~$\tau$, c.f, %
  Figure~\ref{fig:running2}).} %
\label{fig:running3}
\end{figure*}

\shortversion{\vspace{-.25em}}
\begin{example} %
  Recall instance~$(\prog,P)$, TD~$\TTT$, and tables~$\tab{1}$,
  $\ldots$, $\tab{14}$ from Examples~\ref{ex:running0}, \ref{ex:sat},
  and Figure~\ref{fig:running2}. Figure~\ref{fig:running3} depicts
  selected tables of~$\pi_1, \ldots, \pi_{14}$ obtained after
  running~$\dpa_\PROJ$ for counting projected answer sets.
  We assume that row $i$ in table $\pi_t$ corresponds to
  $\vec{v_{t.i}} = \langle \rho_{t.i}, c_{t.i} \rangle$
  where~$\rho_{t.i}\subseteq\nu(t)$.
  Recall that %
  there are rows among
  different~$\PRIM$-tables that are removed (highlighted gray in Figure~\ref{fig:running2}) during purging. %
  By purging we avoid to correct stored counters (backtracking)
  whenever a row has no ``succeeding'' row in the
  parent table.
  \noindent Next, we discuss selected rows obtained by
  $\dpa_\PROJ((\prog,P),(T,\chi,\nu))$. Tables $\pi_1$, $\ldots$,
  $\pi_{14}$ are shown in Figure~\ref{fig:running3}.
  Since~$\type(t_1)= \leaf$, we have
  $\pi_1=\langle\{\langle \emptyset , \emptyset, \langle \rangle
  \rangle \}, 1\rangle$.  Intuitively, at~$t_1$ the
  row~$\langle\emptyset, \emptyset, \langle\rangle\rangle$ belongs to~$1$ bucket.
  Node~$t_2$ introduces atom~$a$, which results in
  table~$\pi_2\eqdef\big\SB\langle \{\vec{u_{2.1}}\},
  1\rangle, \langle \{\vec{u_{2.2}}\},
  1\rangle, \langle \{\vec{u_{2.1}}, \vec{u_{2.2}}\},
  1\rangle\big\SE$, where~$\vec{u_{2.1}}=\langle \emptyset, \emptyset, \langle \rangle\rangle$ and~$\vec{u_{2.2}}=\langle \{a\}, \emptyset, \langle a\rangle \rangle$ 
  (derived similarly to table~$\pi_{4}$ as in Example~\ref{ex:pcnt}). 
  Node~$t_{10}$ introduces projected atom~$e$, and
  node~$t_{11}$ removes~$e$.  
  For row~$\vec{v_{11.1}}$ we compute the
  count~$\ipmc(t_{11},\{\vec{\tabval_{11.1}}\},
  \langle\pi_{10}\rangle)$ by means of~$\pcnt$. Therefore, take
  for~$\varphi$ the singleton set~$\{\vec{\tabval_{11.1}}\}$.
  We simply have
  $\ipmc(t_{11},\{\vec{\tabval_{11.1}}\}, \langle\pi_{10}\rangle) =
  \pmc(t_{11},\{\vec{\tabval_{11.1}}\}, \langle\pi_{10}\rangle)$.  To
  compute
  $\pmc(t_{11},\{\vec{\tabval_{11.1}}\},$ $\langle\pi_{10}\rangle)$, we
  take for~$O$ the sets~$\{\vec{u_{10.1}}\}$, $\{\vec{u_{10.6}}\}$,
  $\{\vec{u_{10.7}}\}$, and~$\{\vec{u_{10.6}}, \vec{u_{10.7}}\}$ into
  account, since all other non-empty subsets of origins
  of~$\vec{\tabval_{11.1}}$ in~$\nu_{10}$ do not occur in~$\pi_{10}$.
  Then, we take the sum over the values
  $\sipmc(\langle \pi_{10}\rangle,\{\vec{\tabval_{10.1}}\})=1$,
  $\sipmc(\langle \pi_{10}\rangle,\{\vec{\tabval_{10.6}}\})=1$,
  $\sipmc(\langle \pi_{10}\rangle,\{\vec{\tabval_{10.7}}\})=1$ and
  subtract
  $\sipmc(\langle \pi_{10}\rangle,$ $\{\vec{\tabval_{10.6}},
  \vec{\tabval_{10.7}}\})=1$. This results
  in~$\pmc(t_{11},\{\vec{\tabval_{11.1}}\}, \langle\pi_{10}\rangle) =
  c_{10.1} + c_{10.7}\; + $ $ c_{10.8} - c_{10.9} = 2$. We proceed similarly
  for row~$v_{11.2}$, resulting in~$c_{11.2}=1$.
  Then for row~$v_{11.3}$,
  $\ipmc(t_{11},\{\vec{\tabval_{11.1}},\vec{\tabval_{11.6}}\}, \langle\pi_{10}\rangle) = | %
  \pmc(t_{11},\{\vec{\tabval_{11.1}},\vec{\tabval_{11.6}}\}, \langle\pi_{10}\rangle) - \ipmc(t_{11},\{$ $\vec{\tabval_{11.1}}\}, \langle\pi_{10}\rangle)$ $- \ipmc(t_{11},\{\vec{\tabval_{11.6}}\}, \langle\pi_{10}\rangle) | = |2-c_{11.1}-$ $c_{11.2}|= |2 -$ $2 - 1| =\Card{-1} = 1 = c_{11.3}$.
  Hence, $c_{11.3} = 1$ represents the number of projected answer sets,
  both rows~$\vec{u_{11.1}}$ and~$\vec{u_{11.6}}$ have in common. We
  then use it for table~$t_{12}$.  Node~$t_{12}$ removes projection
  atom~$d$.  For node~$t_{13}$ where $\type(t_{13}) = \join$ one
  multiplies stored $\sipmc$ values for \AlgA-rows in the two children
  of~$t_{13}$ accordingly.  In the end, the projected answer sets count
  of~$\prog$ is~$\sipmc(\langle\pi_{14}\rangle,\vec{u_{14.1}})=3$.
\end{example}
\noindent%
\FIX{Next, we present upper bounds on the runtime of~$\dpa_{\PROJ}$.  
Therefore, let~$\gamma(n)\in\mathcal{O}(n\cdot log\, n \cdot log\, log\,n)$~\cite{\longversion{Knuth1998,}Harvey2016} be the runtime for multiplying two~$n$-bit integers.}

\shortversion{\begin{theorem}[$\star$]}
\longversion{\begin{theorem}}
  \label{thm:runtime}
  $\dpa_{\PROJ}$ runs in time
  $\mathcal{O}(2^{4m}\cdot g \cdot \gamma(\CCard{\prog}))$
  for instance~$(\prog,P)$ and TTD~$\TTT_{\text{purged}} = (T,\chi,\nu)$ of~$G_\prog$ of width~$k$ with $g$
  nodes,  where~$m\eqdef \max_{t\,\text{in}\,T}(|\nu(t)|)$.
\end{theorem}
\longversion{\begin{proof}
  Let~$d = k+1$ be maximum bag size of the TD~$\TTT$. For each
  node~$t$ of $T$, we consider the table $\nu(t)$ of $\TTT_{\text{purged}}$.
  Let TDD~$(T,\chi,\pi)$ be the output of~$\dpa_\PROJ$. In worst case,
  we store in~$\pi(t)$ each subset~$\rho \subseteq \nu(t)$ together
  with exactly one counter. Hence, we have at most $2^{m}$ many rows
  in $\rho$.
  In order to compute $\ipmc$ for~$\rho$, we consider every
  subset~$\varphi \subseteq \rho$ and compute~$\pcnt$. Since
  $\Card{\rho}\leq m$, we have at most~$2^{m}$ many subsets $\varphi$
  of $\rho$. Finally, for computing $\pcnt$, we consider in the worst
  case each subset of the origins of~$\varphi$ for each child table,
  which are at most~$2^{m}\cdot 2^{m}$ because of nodes~$t$
  with~$\type(t)=\join$.
  In total, we obtain a runtime bound
  of~$\bigO{2^{m} \cdot 2^{m} \cdot 2^{m}\cdot 2^{m} \cdot
    \gamma(\CCard{\prog})} \subseteq \bigO{2^{4m} \cdot
    \gamma(\CCard{\prog}})$ due to multiplication of two $n$-bit
  integers for nodes~$t$ with~$\type(t)=\join$ at costs~$\gamma(n)$.
  Then, we apply this to every node of~$T$ %
  resulting in
  runtime~$\bigO{2^{4m} \cdot g \cdot \gamma(\CCard{\prog})}$.
\end{proof}}

\shortversion{\begin{corollary}[$\star$]}
\longversion{\begin{corollary}}\label{cor:runtime}
  Given an instance $(\prog,P)$ of \PASP where $\prog$ is
  head-cycle-free and~$k=\tw{G_\prog}$. Then, $\mdpa{\PRIM}$ runs in
  time~$\mathcal{O}(2^{3^{k+1.27}\cdot k!}\cdot \CCard{\prog}\cdot
  \gamma(\CCard{\prog}))$.
\end{corollary}
\longversion{\begin{proof}
  We can compute in time~$2^{\mathcal{O}(k^3)}\cdot\CCard{G_\prog}$ a
  TD~${\cal T'}$ with~$g\leq \CCard{\prog}$ nodes of width at
  most~$k$~\cite{Bodlaender96}. Then, we can simply
  run~$\dpa_{\PRIM}$, which runs in
  time~$\mathcal{O}({3^{k}\cdot k!}\cdot \CCard{\prog})$ by
  Theorem~\ref{thm:primruntime} and since the number of nodes of a
  tree decomposition is linear in the size of the input
  instance~\cite{Bodlaender96}. %
  Then, we again traverse the TD for purging and output
  $\TTT_{\text{purged}}$, which runs in time single exponential in the
  treewidth and linear in the instance size. Finally, we run
  $\dpa_{\PROJ}$ and obtain by Theorem~\ref{thm:runtime} that the
  runtime bound
  $\mathcal{O}(2^{4\cdot3^{k}\cdot k!}\cdot \CCard{\prog}\cdot
  \gamma(\CCard{\prog})) \subseteq $
  $\mathcal{O}(2^{3^{k + 1.27}\cdot k!}\cdot \CCard{\prog}\cdot
  \gamma(\CCard{\prog}))$.  %
  Hence, the corollary holds.
\end{proof}}

\longversion{\noindent\FIX{Then, we present lower bounds, and show that~$\mdpa{\PRIM}$ is indeed correct.}} %

\begin{theorem}[Lower Bound\shortversion{, $\star$}]
  Under ETH, $\PASP$ cannot be solved in
  time $2^{2^{o(k)}}\cdot \CCard{\prog}^{o(k)}$ \FIX{for $(\prog,P)$ s.t.~$\prog$ is head-cycle-free, normal or tight, $k=\tw{G_\prog}$.}
\end{theorem}
\longversion{\begin{proof}
  Assume for proof by contradiction that there is such an algorithm.
  We show that this contradicts a very recent
  result~\cite{LampisMitsou17,FichteEtAl18}, which states that one
  cannot decide the validity of a QBF
  $\forall{V_1}.\exists V_2.E$ in
  time~$2^{2^{o(k)}}\cdot \CCard{E}^{o(k)}$,
  where 
  $E$ is in CNF.
  Let $(\forall{V_1}.\exists V_2.E,k)$ be an instance
  of~$\forall\exists$-\SAT parameterized by the treewidth~$k$. Then,
  we reduce to an instance~$((\prog,P),2k)$ of the decision
  version~$\PASP$-exactly-$2^{\Card{V_1}}$ when parameterized by
  treewidth of~$G_\prog$ such that $P=V_1$, the number of solutions is
  exactly~$2^{\Card{V_1}}$, and~$\prog$ is as follows.  For
  each~$v\in V_1 \cup V_2$, program~$\prog$ contains %
  rules~$v \leftarrow \neg nv$ and $nv \leftarrow \neg v$.
  Each clause~$x_1 \vee \ldots \vee x_i \vee \neg x_{i+1} \vee \ldots \vee \neg x_j$
  results in one additional
  rule~$\hsep \neg x_1,\ldots, \neg x_i, x_{i+1}, \ldots, x_{j}$.
  It is easy to see that the reduction is correct and therefore instance~$((\prog,P), 2k)$ is
  a yes instance of %
  $\PASP$-exactly-$2^{\Card{V_1}}$ 
  if and only if~$(\forall{V_1}.\exists V_2.E,k)$
  is a yes instance of %
  problem~$\forall\exists$-\SAT. %
  In fact, $\prog$ is \FIX{head-cycle-free, normal and tight}, and
the reduction \FIX{runs in polynomial time of $\prog$ and at most doubles the treewidth due to duplication of atoms},
  which establishes the result. %
\end{proof}}

\longerversion{
\begin{corollary}
  Unless ETH fails, $\PASP$ cannot be solved in
  time~$2^{2^{o(k)}}\cdot \CCard{\prog}^{o(k)}$ for a given instance
  $(\prog,P)$, where~$k$ is the treewidth of the incidence graph of~$\prog$.
\end{corollary}
\begin{proof}
  Let $w_i$ and $w_p$ be the treewidth of the incidence graph and
  primal graph of~$\prog$, respectively. Then,
  $w_i \leq w_p +1$~\cite{SamerSzeider10b}, which establishes the
  claim.
\end{proof}

\begin{corollary}
  Given an instance $(\prog,P)$ of \PASP where $\prog$ has treewidth~$k$. Then,
  Algorithm~$\mdpa{\AlgA}$ runs in
  time~$2^{2^{\Theta(k)}} \cdot \CCard{\prog}^c$ for some positive
  integer~$c$.
\end{corollary}}

\longversion{Finally, we state that indeed~$\mdpa{\PRIM}$ gives the projected answer sets count of a given head-cycle-free program~$\prog$.

\begin{proposition}[$\star$]\label{prop:phcworks}
  Algorithm $\mdpa{\PRIM}$ is correct and outputs for any instance
  of \PASP restricted to head-cycle-free programs its projected answer sets count.
\end{proposition}
\begin{proof}[Proof (Idea)]
Soundness follows by establishing an invariant for any row of~$\pi(t)$ guaranteeing that the values of~$\ipmc$ indeed capture ``all-overlapping'' counts of~$\progt{t}$. One can show that the invariant is a consequence of the properties of~\PRIM and the additional ``purging'' step, which neither destroys soundness nor completeness of~$\dpa_\PRIM$. Further, completeness guarantees that all required rows are computed.
\end{proof}}
\noindent\textbf{Solving \PDASP for Disjunctive Programs.}
\FIX{We extend our algorithm to projected answer set
  counting for disjunctive programs.} Therefore, we
simply use a table algorithm \AlgA=\algo{PRIM} for disjunctive ASP
\shortversion{as in previous literature~\cite{FichteEtAl17a}}%
\longversion{%
  that was introduced in the
  literature~\cite{FichteEtAl17a,JaklPichlerWoltran09}}.
Recall algorithm~\mdpa{\AlgA}\longversion{ illustrated in
Figure~\ref{fig:multiarch}}.
First, we %
heuristically compute a TD of the primal graph. Then, we run $\dpa_\algo{PRIM}$ as first
traversal resulting in TTD~$(T,\chi,\tau)$. Next, we purge rows
of~$\tau$, which cannot be extended to an answer set resulting in
TTD~$(T,\chi,\nu)$. Finally, we use~$(T,\chi,\nu)$ to compute the projected answer sets count
by~$\dpa_{\PROJ}$\shortversion{.}\longversion{ and obtain TTD~$(T,\chi,\pi)$.}
\longversion{
\begin{proposition}[$\star$]\label{prop:disjworks}
  $\mdpa{\algo{PRIM}}$ is correct, i.e., it\shortversion{ solves~\PDASP}\longversion{ outputs the projected answer sets count for any instance
  of \PDASP}.
\end{proposition}}
\begin{table}[t]
\centering
\begin{tabular}[t]{l@{\hspace{0.5em}}c@{\hspace{0.5em}}||@{\hspace{0.5em}}r@{\hspace{0.5em}}|@{\hspace{0.5em}}r@{\hspace{0.5em}}}
	Problem & Restriction & Upper Bound & Lower Bound (under ETH)\\\hline
	\SAT, \cSAT & - & $2^{\bigO{k}}\cdot poly(\CCard{\prog})$~\cite{SamerSzeider10b} & $2^{\Omega(k)}\cdot poly(\CCard{\prog})$~\cite{ImpagliazzoPaturiZane01} \\
	\ASP, \cASP & tight & $\mathbf{{2^{\bigO{k}}}\cdot poly(\CCard{\prog})}$ & ${2^{\Omega(k)}}\cdot poly(\CCard{\prog})$~\cite{ImpagliazzoPaturiZane01} \\
	\ASP, \cASP & normal, HCF & $\mathbf{{2^{\bigO{k\cdot log(k)}}}\cdot poly(\CCard{\prog})}$ & ${2^{\Omega(k)}}\cdot poly(\CCard{\prog})$~\cite{ImpagliazzoPaturiZane01} \\
	\ASP, \cASP & disjunctive & $2^{2^{\bigO{k}}}\cdot poly(\CCard{\prog})$~\cite{JaklPichlerWoltran09} & ${2^{2^{\Omega(k)}}\cdot poly(\CCard{\prog})}$\shortversion{~\cite{FichteEtAl17a}}\longerversion{~\cite{HecherFichte19}} \\
	Proj.\ \cSAT & - & $2^{2^{\bigO{k}}}\cdot poly(\CCard{\prog})$~\cite{FichteEtAl18} & $2^{2^{\Omega(k)}}\cdot poly(\CCard{\prog})$~\cite{FichteEtAl18} \\
	\PASP & tight & $\mathbf{2^{2^{\bigO{k}}}\cdot poly(\CCard{\prog})}$ & $\mathbf{2^{2^{\Omega(k)}}\cdot poly(\CCard{\prog})}$ \\
	\PASP & normal, HCF & $\mathbf{2^{2^{\bigO{k\cdot log(k)}}}\cdot poly(\CCard{\prog})}$ & $\mathbf{2^{2^{\Omega(k)}}\cdot poly(\CCard{\prog})}$ \\
	\PDASP & disjunctive & $\mathbf{2^{2^{2^{\bigO{k}}}}\cdot poly(\CCard{\prog})}$ & $\mathbf{2^{2^{2^{\Omega(k)}}}\cdot poly(\CCard{\prog})}$ \\
\end{tabular}\vspace{0.4em}
\caption{Overview of upper and lower bounds using treewidth~$k$ of the primal graph of instance~$\prog$; bold entries were established in the course of this paper.}
\shortversion{\vspace{-1.8em}}
\label{tbl:summary}
\end{table}
\longversion{The following lemma states the runtime results.}
\shortversion{\vspace{-.4em}
\begin{lemma}[$\star$]}
\longversion{\begin{lemma}}
\label{cor:disjruntime}
$\mdpa{\algo{PRIM}}$
  runs in
  time~$\mathcal{O}(2^{2^{2^{k+3}}}\cdot \CCard{\prog}\cdot
  \gamma(\CCard{\prog}))$ for given instance $(\prog,P)$ of \PDASP where $\prog$ is a
  disjunctive program, and $k=\tw{G_\prog}$.
\end{lemma}
\longversion{%
\begin{proof}
  The first two steps follow the proof of Corollary~\ref{cor:runtime}.
  However, $\dpa_{\algo{PRIM}}$ runs in
  time~$\mathcal{O}(2^{2^{k+2}}\cdot
  \CCard{\prog})$~\cite{FichteEtAl17a}. Finally, we run $\dpa_{\PROJ}$
  and obtain by Theorem~\ref{thm:runtime} that
  $\mathcal{O}(2^{4\cdot2^{2^{k+2}}}\cdot \CCard{\prog}\cdot
  \gamma(\CCard{\prog})) \subseteq $
  $\mathcal{O}(2^{2^{2^{k+3}}}\cdot \CCard{\prog}\cdot
  \gamma(\CCard{\prog}))$.  
\end{proof}
}

\longversion{Then, the runtime of algorithm~$\mdpa{\algo{PRIM}}$ cannot be significantly improved.} %

\longerversion{\begin{proposition}
  Unless ETH fails, QBFs of the form $\exists V_1.\forall V_2.\cdots\forall V_\ell. E$
where~$k$ is the treewidth of the primal graph of DNF formula~$E$, cannot be solved in time~$2^{2^{\dots^{2^{o(k)}}}}\cdot \CCard{\prog}^{o(k)}$,
where the height of the tower is~$\ell$.
\end{proposition}

\begin{proof}[Proof (Idea)]
	Follows by construction defined in the proof of~\cite{LampisMitsou17} for QBFs of the form~$\exists V_1.\forall V_2. E$. \todo{provide rigorous proof?}
\end{proof}

\begin{corollary}
Unless ETH fails, QBFs of the form $\forall V_1.\exists V_2.\cdots\forall V_\ell. E$
where~$k$ is the treewidth of the primal graph of CNF formula~$E$, cannot be solved in time~$2^{2^{\dots^{2^{o(k)}}}}\cdot \CCard{\prog}^{o(k)}$,
where the height of the tower is~$\ell$.
\end{corollary}}

\shortversion{\vspace{-1.1em}}
\begin{theorem}[Lower Bound\shortversion{, $\star$}]
\label{thm:lowerbound_disj}
  \PDASP cannot be
  solved in time~$2^{2^{2^{o(k)}}} \cdot \CCard{\prog}^{o(k)}$ for
  given instance~$(\prog, P)$, where~$k=\tw{G_\prog}$, unless ETH fails. 
\end{theorem}
\longversion{\begin{proof}
  Assume for proof by contradiction that there is such an algorithm.
  We show that this contradicts a rather recent result~\cite{FichteHecherPfandler19}
  stating that %
  one cannot decide validity of QBF
  $Q=\forall{V_1}.\exists V_2.\forall V_3.E$ in
  time~$2^{2^{2^{o(k)}}}\cdot \CCard{E}^{o(k)}$
  where %
  $E$ is in DNF\FIX{, which was anticipated by Marx and Mitsou~\cite{MarxMitsou16}}. 
  Assume we have
  given such an instance when parameterized by the treewidth~$k$.
  In the following, we employ a well-known
  reduction~$R$~\cite{EiterGottlob95}, which
  transforms~$\exists V_2.\forall V_3. E$
  into~$\prog=R(\exists V_2.\forall V_3. E)$ and gives a yes
  instance~$\prog$ of consistency if and only
  if~$\exists V_2. \forall V_3. E$ is a yes instance of
  $\exists\forall$-\SAT.
  Then, we reduce instance~$(Q,k)$ via a reduction~$S$ to an
  instance~$((\prog',V_1),2k+2)$, where $\prog'=R(\exists V_2'.\forall V_3. E)$, $V_2'\eqdef V_1\cup V_2$, of the decision
  version~$\PDASP$-exactly-$2^{\Card{V_1}}$ of~$\PDASP$ when
  parameterized by treewidth such that the number of projected answer
  sets is
  exactly~$2^{\Card{V_1}}$.
  It is easy to see that reduction~$S$
  gives a yes instance~$(\prog',V_1)$
  of~$\PDASP$-exactly-$2^{\Card{V_1}}$ if and only
  if~$\forall V_1.\exists V_2. \forall V_3. E$ is a yes instance
  of~$\forall\exists\forall$-\SAT.
  However, it remains to show that the reduction~$S$ indeed increases
  the treewidth only linearly.
  Therefore, let $\TTT=(T,\chi)$ be TD of~$E$. We transform~$\TTT$
  into a TD~$\TTT'=(T,\chi')$ of~$G_{\prog'}$ as follows.  For each
  bag~$\chi(t)$ of~$\TTT$, we add vertices for the atoms~$w$ and $w'$
  (two additional atoms introduced in reduction~$R$) and in addition
  we duplicate each vertex~$v$ in~$\chi(t)$ (due to corresponding duplicate
  atoms introduced in reduction~$R$). Observe
  that~$\width(\TTT') \leq 2\cdot \width(\TTT) + 2$. By construction
  of~$R$, $\TTT'$ is then a TD of~$G_{\prog'}$.
  Hence, $S$ runs in polynomial time and linearly preserves the parameter. %
\end{proof}}

\longerversion{
\begin{corollary}
Unless ETH fails, \PDASP for disjunctive programs~$\prog$ cannot be solved in time~$2^{2^{2^{o(k)}}} \cdot \CCard{\prog}^{o(k)}$ for given instance
~$(\prog, P)$, where~$k$ is the treewidth of the incidence graph of~$\prog$.
\end{corollary}}

\shortversion{\vspace{-0.1em}}
\noindent\FIX{In total, we obtain results presented in Table~\ref{tbl:summary}.
Indeed, there is an increase of complexity when going from~\ASP and~\cASP to~\PASP (c.f., Theorem~\ref{thm:runtime}).
For solving \ASP (\cASP) on tight programs one can again reuse Algorithm~\algo{PHC} (Listing~\ref{fig:prim})
without the orderings~$\sigma$, or encode~\cite{Fages94} to~\SAT and use established DP algorithms~\cite{SamerSzeider10b} for~\SAT (\cSAT).
Then, \PASP on tight programs can be solved after purging, followed by computing projected answer sets by means of~$\dpa_{\PROJ}$.
}

\longerversion{
In fact, we can conclude the following
corollary, which renders algorithm~$\mdpa{\algo{PRIM}}$ asymptotically
worst-case optimal, depending on the costs~$\gamma(n)$ for multiplying
two~$n$-bit numbers.

\begin{corollary}
Unless ETH fails, \PDASP for disjunctive programs~$\prog$ runs in time~$2^{2^{2^{\Theta(k)}}} \cdot \CCard{\prog} \cdot \gamma(\CCard{\prog})$ for given instance
~$(\prog, P)$, where~$k$ is the treewidth of the primal graph of~$\prog$.
\end{corollary}
\begin{proof}
Lower bounds by Theorem~\ref{thm:lowerbound_disj}. Upper bound by Algorithm~$\dpa_{\algo{PRIM}}$ in the first pass (c.f., Lemma~\ref{cor:disjruntime}),
followed by purging and the projection algorithm~$\dpa_\PROJ$ in the second pass.
\end{proof}
}
\shortversion{\vspace{-.75em}}
\section{Conclusions}\label{sec:conclusions}\shortversion{\vspace{-.5em}}
We introduced novel algorithms to count the projected answer sets
(\PASP) of tight, normal, head-cycle-free, and \FIX{arbitrary} disjunctive programs. Our algorithms
employ dynamic programming and exploit small treewidth of the
primal graph of the input program. 
More precisely, for disjunctive programs, the runtime is triple exponential in the treewidth and
polynomial in the size of the instance, which can not
be significantly improved under the exponential time hypothesis.
When we restrict the
input to tight, normal, and head-cycle-free programs, the runtime drops to double
exponential, c.f., Table~\ref{tbl:summary}.
Our results extend previous work to
answer set programming and we believe it is applicable to
further hard combinatorial problems, such as 
quantified Boolean formulas\longversion{(QBF)~\cite{CharwatWoltran16a}}
and circumscription~\cite{DurandHermannKolaitis05}.

\shortversion{
}

\longversion{
\bibliography{references}}

\appendix
\longversion
{
\clearpage
\section{Additional Resources}

\appendix

\subsection{Additional Examples}
\begin{example}[c.f.,\cite{FichteEtAl17a}]\label{ex:bagprog} 
  Intuitively, the tree decomposition of Figure~\ref{fig:graph-td}
  enables us to evaluate program $\prog$ by analyzing sub-programs
  $\{r_2\}$ and $\{r_3,r_4, r_5\}$, and combining results agreeing on
  $e$ followed by analyzing~$\{r_1\}$.  Indeed, for the given tree
  decomposition of Figure~\ref{fig:graph-td}, $\progt{t_1}=\{r_2\}$,
  $\progt{t_2}=\{r_3,r_4, r_5\}$ and
  $\prog=\progt{t_3}=\{r_1\} \cup \progtneq{t_3}$. Note that
  here~$\prog=\progt{t_3} \neq \progtneq{t_3}$ and the tree
  decomposition is not nice.  \longerversion{For the tree decomposition
    of Figure~\ref{fig:graph-td2}, we have
    $\progt{t_1} = \{r_1,r_2\}$, %
    as well as $\progt{t_3} = \{r_3\}$.} %
\end{example}%

\subsection{Parsimonious reductions}
Let $L$ and $L'$ be counting problems with witness functions~$\WWW$
and $\WWW'$. A \emph{parsimonious reduction} from~$L$ to $L'$ is a
polynomial-time reduction~$r: \Sigma^* \rightarrow \Sigma'^*$ such
that for all~$I \in \Sigma^*$, we
have~$\Card{\WWW(I)}=\Card{\WWW'(r(I))}$. It is easy to see that the
counting complexity classes~$\cntc\mtext{C}$ defined above are closed
under parsimonious reductions. It is clear for counting problems~$L$
and $L'$ that if $L \in \cntc\mtext{C}$ and there is a parsimonious
reduction from~$L'$ to $L$, then $L' \in \cntc\mtext{C}$.

\subsection{Counting Complexity of~\PASP: Omitted proofs}
\begin{restatetheorem}[prop:hcfproj]
\begin{theorem}
  The problem~$\PASP$ is \sharpSigma{2}-complete when we allow
  disjunctive programs as input and \sharpNP-complete when the input
  is restricted to \FIX{head-cycle-free, normal or tight programs.}
\end{theorem}
\end{restatetheorem}
\begin{proof}
  Membership immediately holds as we can check for a given
  set~$I\subseteq P$ whether there is an answer set~$J\supseteq I$
  of~$\prog$ with~$J\cap (P\setminus I)=\emptyset$ by checking if
  there is an answer set of
  program~$\prog \cup \bigcup_{i\in I} \{\hsep \neg i\} \cup
  \bigcup_{i\in P\setminus I} \{\hsep i\}$. \FIX{Note that if~$\prog$ is
  head-cycle-free, normal, or tight, this program is again head-cycle-free, normal, or tight, respectively.} Hardness
  follows by establishing a parsimonious reduction
  from~$\cnt\exists$-\SAT or $\cnt\exists\forall$-\SAT\footnote{For
    quantified Boolean formulas (QBF) and its evaluation problem
    ($Q_1\ldots Q_i$-\SAT for alternating
    $Q_i \in \{\exists, \forall\}$) we refer to standard
    texts~\cite{BiereHeuleMaarenWalsh09\longversion{,KleineBuningLettman99}}.},
  respectively.
  Assume that the input is restricted to \FIX{head-cycle-free, normal or tight}
  programs. 
  Given an instance~$(Q, Z)$ with $Q=\exists X.\phi(X,Z)$. We reduce
  to the instance~$(R(Q), Z)$ of $\PASP$, where~$R(Q)$ is defined as
  follows. For each variable~$v\in X\cup Z$, we add \FIX{the
  rules~$v \leftarrow \neg nv$ and $nv \leftarrow \neg v$}. For each clause~$\ell_1 \vee \ldots \vee \ell_k$
  in~$\phi(X,Z)$, we add a
  rule~$\hsep \bar{\ell_1}, \ldots, \bar{\ell_k}$ where~$\bar{\ell_i}$
  corresponds to $x$ if~$\ell_i=\neg x$ for a variable~$x$, and $\neg x$ otherwise. Then, a
  counter~$c$ solves~$(Q, Z)$ if and only if $c$ solves~$(R(Q),
  Z)$.
  Assume that we allow arbitrary disjunctive programs as input.  Given
  an instance~$(Q, Z)$, where $Q=\exists X. \forall Y. \phi(X,Y,Z)$. We
  reduce to the instance~$(R(Q'), Z)$ of $\PASP$, where~$Q'=\exists X'. \forall Y.\phi(X,Y,Z)$, $X'=X\cup Z$, and~$R(Q')$ is
  defined exactly as by Eiter and Gottlob~\cite{EiterGottlob95}. Then, since $R$ is a
  correct encoding of $\exists\forall$-\SAT, the projected model count~$c$
  of~$(Q, Z)$ is the projected answer sets count of~$(R(Q'), Z)$ and
  vice versa. Consequently, the proposition sustains.
\end{proof}

\longversion{
\subsection{Worst-Case Analysis of $\dpa_{\PRIM}$: Omitted proofs}

\begin{restatelemma}[prop:kfact]
\begin{lemma}
Given any fixed positive integer~$c \geq 1$ and
  functions~$f(k)\eqdef k!, g(k) \eqdef$ 
  $2^{k^{(c+1)/c}}$, where~$k$ is a non-negative integer. Then,
  $f \in \bigO{g}$.
\end{lemma}
\end{restatelemma}
\begin{proof}
We proceed by simultaneous induction.\\
Base case ($k=c=1$): Obviously, $1^{2} \geq 1!$.\\
Induction hypothesis: $k! \in O(2^{k^{(c+1)/c}})$\\
Induction step ($k \rightarrow k+1$): \\We have to show that for $k\geq k_0$ for some fixed $k_0$, the following equation holds.
\begin{align*}
  2^{(k+1)^{(c+1)/c}} \geq (k+1)\cdot k!\\
  2^{(k+1)^{1/c}\cdot(k+1)} \geq (k+1)\cdot k!\\
  2^{(k+1)^{1/c}+k\cdot(k+1)^{1/c}} \geq (k+1)\cdot k!\\
  2^{(k+1)^{1/c}}\cdot 2^{k\cdot(k+1)^{1/c}} \geq (k+1)\cdot k!\\
  2^{(k+1)^{1/c}} \cdot k! \geq^{IH} (k+1)\cdot k!\\
  2^{(k+1)^{1/c}} \geq (k+1)\\
  2^{(k+1)^{1/c}} \geq 2^{\text{log}_2(k+1)}\geq (k+1)\\
  \text{ where } k\geq k_0 \text{ for some fixed } k_0 \text{ since } \text{log}_2\in O(\text{exp}(1/c))
\end{align*}
Induction step ($k \rightarrow k+1, c \rightarrow c+1$): Analogous, previous step works for any~$c$.\\
Induction step ($c \rightarrow c+1$): Analogous.
\end{proof}}

\subsection{Characterizing Extensions}

In the following, we assume~$(\prog,P)$ to be an instance of~$\PASP$. Further, let~$\mathcal{T}=(T,\chi,\tau)$
be an~$\AlgA$-TTD of~$G_\prog$ where~$T=(N,\cdot,n)$, node~$t\in N$, and~$\rho\subseteq\tau(t)$.

\begin{definition}\label{def:extensions}
  Let $\vec u$ be a row of $\rho$.

  An \emph{extension below~$t$} is a set of pairs where a pair consist
  of a node~$t'$ of the \emph{induced sub-tree~$T[t]$ rooted at~$t$} and a row~$\vec v$ of $\tau(t')$
  and the cardinality of the set equals the number of nodes in the
  sub-tree~$T[t]$. 
  
  We define the family of \emph{extensions below~$t$}
  recursively as follows.  If $t$ is of type~\leaf, then
  $\Ext_{\leq t}(\vec u) \eqdef \{\{\langle t,\vec u\rangle\}\}$;
  otherwise
  $\Ext_{\leq t}(\vec u) \eqdef \bigcup_{\vec v \in \origs(t,\vec u)}
  \big\SB$ $\{\langle t,\vec u\rangle\}\cup X_1 \cup \ldots \cup X_\ell
  \SM X_i\in\Ext_{\leq t_i}({\vec v}_{(i)})\big\SE$ %
  for the~$\ell$ children~$t_1, \ldots, t_\ell$ of~$t$.
  We extend this notation for an $\AlgS$-table~$\rho$ by
  $\Ext_{\leq t}(\rho)\eqdef \bigcup_{\vec u\in\rho} \Ext_{\leq
    t}(\vec u)$.  Further, we
  let~$\Exts \eqdef \Ext_{\leq n}(\tau(n))$ be the
  \emph{family of all extensions}. 
  
  Further, we define \emph{the local table for node}~$t$ and family~$E$ of extensions (below some node) as
  $\local(t,E)\eqdef \bigcup_{\hat\rho \in E}\{ \langle \vec{\tabval}\rangle \mid
  \langle t, \vec{\tabval}\rangle \in \hat{\rho}\}$.

\end{definition}

If we would construct all extensions below the root~$n$, it allows us
to also obtain all models of program~$\prog$.  To this end, we state the following definition.

\begin{definition}\label{def:satext}
  We define %
  the \emph{satisfiable
    extensions below~$t$} for~$\rho$ by
  \[\PExt_{\leq t}(\rho)\eqdef \bigcup_{\vec u\in\rho} \SB X \SM X
    \in \Ext_{\leq t}(\vec u), X \subseteq Y, Y \in \Exts\SE.\]
\end{definition}

\begin{observation}
$\PExt_{\leq n}(\tau(n)) = \Exts$.
\end{observation}

\begin{definition}
We define the \emph{purged table mapping~$\nu$ of~$\tau$} by
$\nu(t)\eqdef \local(t,$ $\PExt_{\leq t}[\tau(t)])$ for every~$t\in N$.
\end{definition}

Next, we define an auxiliary notation that gives us a way to
reconstruct interpretations from families of extensions.

\begin{definition}\label{def:iextensions}
  Let $E$ be a family of extensions
  below~$t$. %
  We define the \emph{set~$\mathcal{I}(E)$ of interpretations} of~$E$
  by
  $\mathcal{I}(E) \eqdef \big\SB$ $\bigcup_{\langle \cdot, \vec u
    \rangle \in X} \mathcal{I}(\vec u) \mid X \in E \big\SE$
  and the set~$\mathcal{I}_P(E)$ of \emph{projected interpretations} by
  $\mathcal{I}_P(E) \eqdef \big\SB \bigcup_{\langle \cdot, \vec u \rangle \in X}
  \mathcal{I}(\vec u) \cap P \mid X \in E \big\SE$.

\end{definition}

\begin{example} %
  Consider again program~$\prog$ and TTD~$(T,\chi,\tau)$ of~$G_\prog$,
  where~$t_{14}$ is the root of~$T$, from Example~\ref{ex:sat}.
  Let~$X=\{\langle t_{13}, \langle\{b\}, \{b\}, \langle b\rangle\rangle\rangle, \langle t_{12},
  \langle\{b\}, \emptyset, \langle b\rangle\rangle\rangle,
  \langle t_{11},$
  $\langle\{b\}, \emptyset, \langle b\rangle\rangle\rangle,$
  $\langle t_{10},
  \langle\{b,e\}, \{e\}, \langle b,e\rangle\rangle\rangle,$
  $\langle t_{9},
  \langle\{e\}, \{e\}, \langle e\rangle\rangle\rangle,$
  $\langle t_{4},
  \langle\{b\}, \{b\},$ 
  $\langle b\rangle\rangle\rangle,
  \langle t_{3},$
  $\langle\{b\}, \{b\}, \langle b\rangle\rangle\rangle,
  \langle t_{1},
  \langle\emptyset, \emptyset, \langle \rangle\rangle\rangle\}$
  be an extension
  below~$t_{14}$.  Observe that~$X\in\Exts$ and that
  Figure~\ref{fig:running2} highlights those rows of tables for
  nodes~$t_{13}, t_{12}, t_{11}, t_{10}, t_{9}, t_4, t_3$ and~$t_1$ that also occur in~$X$
  (in yellow). Further, $\mathcal{I}(\{X\})=\{b,e\}$ computes the
  corresponding answer set of~$X$, and $\mathcal{I}_P(\{X\}) = \{e\}$ derives
  the projected answer sets of~$X$.  $\mathcal{I}(\Exts)$ refers to the set
  of answer sets of~$\prog$, whereas~$\mathcal{I}_P(\Exts)$ is the set
  of projected answer sets of~$\prog$.
\end{example}

\subsection{Correctness of~$\dpa_{\PRIM}$: Omitted proofs}

In the following, we assume~$\prog$ to be a head-cycle-free program. Further, let~$\mathcal{T}=(T,\chi,\tau)$
be an~$\AlgA$-TTD of~$G_\prog$ where~$T=(N,\cdot,n)$ and~$t\in N$ is a node.

We state definitions required for the correctness
proofs of our algorithm \PRIM. In the end, we only store rows that
are restricted to the bag content to maintain runtime bounds. 
Similar to related work~\cite{FichteEtAl17a}, we define the
content of our tables in two steps. First, we define the properties of
so-called \emph{$\PRIM$-solutions up to~$t$}. Second, we restrict
these solutions to~\emph{$\PRIM$-row solutions} at~$t$.

\begin{definition}\label{def:globalhcf}
Let~$\hat I\subseteq\att{t}$ be an interpretation,
$\hat{\mathcal{P}}\subseteq \hat I$ be a set of atoms and~$\hat\sigma$ be an ordering over atoms~$\hat I$.
Then, $\langle \hat I, \hat{\mathcal{P}}, \hat\sigma\rangle$ is referred to as~\emph{$\PRIM$-solution up to~$t$} if the following holds.
  \begin{enumerate}
    \item~$\hat I\models\progt{t}$,
    \item for each~$a\in\hat I\cap\attneq{t}$, we have~$a\in\hat{\mathcal{P}}$, and
    \item $a\in\hat{\mathcal{P}}$ if and only if~$a$ is proven using program~$\progt{t}$ and ordering~$\hat\sigma$.
  \end{enumerate}
\end{definition}

Next, we observe that the $\PRIM$-solutions up to~$n$ suffice to capture all the answer sets.

\begin{proposition}\label{prop:hcfglobal}
The set of~$\PRIM$-solutions up to~$n$ characterizes the set of answer sets of~$\prog$.
In particular: $\{\hat I \mid  \langle \hat I, \hat{\mathcal{P}}, \hat\sigma \rangle \text{ is a } \PRIM\text{-solution up to }n\} = \{I \mid I \text{ is an answer set of }\prog\}$.
\end{proposition}
\begin{proof}
Observe that Definition~\ref{def:globalhcf} for root node~$t=n$ indeed suffices for~$\hat I$ to be a model of~$\progt{n}=\prog$,
and, moreover, every atom in~$\hat I=\hat P$ is proven in~$\prog$ by ordering~$\hat\sigma$.
\end{proof}

\begin{definition}\label{def:localhcf}
Let~$\langle \hat I, \hat{\mathcal{P}}, \hat\sigma\rangle$ be a~$\PRIM$-solution up to~$t$. Then, $\langle \hat I \cap \chi(t), \hat{\mathcal{P}} \cap \chi(t), \sigma \rangle$, where~$\sigma$ is the partial ordering of~$\hat\sigma$ only containing~$\chi(t)$, is referred to as~\emph{$\PRIM$-row solution at node~$t$}.
\end{definition}

Given a~$\PRIM$-solution~$\vec{\hat\tabval}$ up to~$t$ and a~$\PRIM$-row solution~$\vec\tabval$ at~$t$.
We say~$\vec{\hat\tabval}$ is a \emph{corresponding} $\PRIM$-solution up to~$t$ of~$\PRIM$-row solution at~$t$ if~$\vec{\hat\tabval}$ can be used
to construct~$\vec\tabval$ according to Definition~\ref{def:localhcf}.

In fact,~\emph{$\PRIM$-row solutions} at~$t$ suffice to capture all the answer sets of~$\prog$.
Before we show that, we need the following definition.

\begin{definition}
Let $t\in N$ be a node of~$\TTT$
  with~$\children(t,T) = \langle t_1, \ldots, t_\ell \rangle$.
Further, let~$\vec{\hat\tabval}=\langle \hat I, \hat{\mathcal{P}}, \hat\sigma\rangle$ be a~$\PRIM$-solution up to~$t$ 
and~$\vec{\hat v}=\langle \hat{I'}, \hat{\mathcal{P}'}, \hat{\sigma'} \rangle$ be a~$\PRIM$-solution up to
$t_i$. Then,~$\vec\tabval$ is \emph{compatible with~$\vec v$} (and vice-versa) if
	\begin{enumerate}
		\item $\hat{I'} = \hat{I}\cap \att{t_i}$
		\item $\hat{\mathcal{P}'} = \hat{\mathcal{P}}\cap \att{t_i}$
		\item $\hat{\sigma'}$ is a sub-sequence of~$\hat\sigma$ such that~$\hat\sigma$ may additionally contain atoms in~$\att{t}\setminus\att{t_i}$
	\end{enumerate}
\end{definition}

\begin{lemma}[Soundness]\label{lem:paspcorrect}
  Let $t\in N$ be a node of~$\TTT$
  with~$\children(t,T) = \langle t_1, \ldots, t_\ell \rangle$.
  Further, let $\vec v_i$ be a~$\PRIM$-row solution at~$t_i$ for~$1\leq i\leq \ell$.
  Then, each row~$\vec\tabval = \langle I, \mathcal{P}, \sigma \rangle$ in~$\tau(t)$ 
  with~$\langle \vec v_1, \ldots, \vec v_\ell \rangle \in \origa{\PRIM}(t, \vec\tabval)$ is also a~$\PRIM$-row solution at
  node~$t$. Moreover, for any corresponding~$\PRIM$-solution~${\vec{\hat\tabval}}$ up to~$t$ (of~$\vec\tabval$)
  there are corresponding \emph{compatible}~$\PRIM$-solutions~$\vec{\hat{v_i}}$ up to~$t_i$ (for~$\vec v_i$). %
\end{lemma}
\begin{proof}[Proof (Sketch)]
We proceed by case distinctions.
Assume case(i):~$\type(t)=\leaf$. Then, $\langle \emptyset, \emptyset, \langle \rangle \rangle$ is a~\PRIM-row solution at~$t$. This concludes case(i).

Assume case(ii):~$\type(t)=\intr$ and~$\chi(t)\setminus\chi(t')=\{a\}$. Let~$\vec v_1=\langle I, \mathcal{P}, \sigma\rangle$ be any \PRIM-row solution at child node~$t_1$,
and~$\vec{\hat{v_1}}=\langle \hat I, \hat{\mathcal{P}}, \hat\sigma\rangle$ be any corresponding \PRIM-solution up to~$t_1$, which exists by Definition~\ref{def:localhcf}. In the following, we show that the way~\PRIM transforms
\PRIM-row solution~$\vec v_1$ at~$t_1$ to a \PRIM-row solution~$\vec\tabval=\langle I', \mathcal{P}', \sigma'\rangle$ at~$t$ is sound.
We identify several sub-cases.

Case (a): Atom~$a\not\in I'$ is set to false. Then, \PRIM constructs~$\vec\tabval$ where~$I'=I, \sigma'=\sigma$ and~$\mathcal{P}'=\mathcal{P}\cup \gatherproof(I',\sigma', \prog_t)$. Note that by construction~$I'\models \prog_t$.
Towards showing soundness, we define how to transform~$\vec{\hat{v_1}}$ into~$\vec{\hat\tabval}$ such that~$\vec{\hat\tabval}$ is indeed the corresponding~$\PRIM$-solution up to~$t$ of row~$\vec{\tabval}$ constructed by~\PRIM. To this end, we define~$\vec{\hat\tabval}$ as follows: $\vec{\hat\tabval} = \langle \hat I, \hat{\mathcal{P}} \cup  \gatherproof(I',\sigma', \prog_t), \hat\sigma\rangle$. Observe that~$\vec{\hat\tabval}$ is a~\PRIM-solution up to~$t$ according to Definition~\ref{def:globalhcf}.
Moreover, by construction and Definition~\ref{def:localhcf}, $\vec{\hat\tabval}$ is a corresponding~$\PRIM$-solution up to~$t$ of~$\hat\tabval$. 
It remains to show, that indeed for any
corresponding~$\PRIM$-solution~${\vec{\hat\tabval}}=\langle \hat {I'},
\hat{\mathcal{P}'}, \hat{\sigma'} \rangle$ up to~$t$
(of~$\vec\tabval$, there is a
corresponding~$\PRIM$-solution~$\vec{\hat{\zeta_1}}$ up to~$t_1$
(of~$\vec{{v_1}}$). %
To this end, we
define~$\vec{\hat{\zeta_1}}=\langle \hat{I'}, \hat{\mathcal{P}'}
\setminus (\mathcal{P}'\setminus\mathcal{P}), \hat{\sigma'}\rangle$
that is by construction according to Definition~\ref{def:globalhcf}
indeed a corresponding~$\PRIM$-solution up to~$t_1$
of~$\vec{\hat{v_1}}$.  This concludes case (a).

Case (b): Atom~$a\in I'$ is set to true. Conceptually, the case works analogously. %
This concludes cases (b) and (ii). 

The remaining cases for nodes~$t$ with~$\type(t)=\rem$ (slightly easier) and nodes~$t$ with~$\type(t)=\join$, 
where we need to consider \PRIM-row solutions at two different child nodes of~$t$, go through similarly.
\end{proof}

\begin{lemma}[Completeness]\label{lem:primcomplete}
  Let~$t\in N$ be node of~$\TTT$ where
  $\type(t) \neq \leaf$ and~$\children(t,T) = \langle t_1, \ldots, t_\ell \rangle$. Given a
  $\PRIM$-row solution~$\vec\tabval=\langle I, \mathcal{P}, \sigma \rangle$ at node~$t$,
  and any corresponding~$\PRIM$-solution~$\vec{\hat\tabval}$ up to~$t$ (of~$\vec\tabval$).
  Then, there exists $\vec s=\langle {v_1}, \ldots, {v_\ell}\rangle$ where ${v_i}$ is a
  $\PRIM$-row solution at~$t_i$ %
  such that~$\vec s\in\origa{\PRIM}(t,\vec\tabval)$,
  and corresponding~$\PRIM$-solution~$\vec{\hat{v_i}}$ up to~$t_i$ (of~$v_i$) that is
  compatible with~$\vec{\hat\tabval}$.
\end{lemma}
\begin{proof}[Proof (Idea)]
Since~$\vec\tabval$ is a~\PRIM-row solution at~$t$, there is by Definition~\ref{def:localhcf} a corresponding~\PRIM-solution~$\vec{\hat\tabval}=\langle \hat I, \hat{\mathcal{P}}, \hat\sigma\rangle$ up to~$t$. 

We proceed again by case distinction. Assume that~$\type(t)=\intr$. Then we define~$\vec{\hat{v_1}}\eqdef \langle \hat I \setminus \{a\}, \hat{\mathcal{P}'}, \hat{\sigma'}\rangle$,
where~$\hat{\sigma'}$ is a sub-sequence of~$\hat\sigma$ that does not contain~$a$ and~$\hat{\mathcal P}'=\gatherproof(\hat I \setminus \{a\}, t_1, \progt{t_1})$. 
Observe that all the conditions of Definition~\ref{def:globalhcf} are met and that~$\hat{\mathcal P}'\subseteq \hat{\mathcal{P}'}$. Then, we can easily define \PRIM-row solution~$\vec{v_1}$ at~$t_1$ according to Definition~\ref{def:localhcf} by using~$\vec{\hat{v_1}}$. By construction of~$\vec{\hat{v_1}}$ and by the definition of~$\gatherproof$, we conclude that~$\vec\tabval$ can be constructed with~$\PRIM$
using~$\vec{v_1}$. Moreover, \PRIM-solution~$\vec{\hat{v_1}}$ up to~$t_1$ is indeed compatible with~$\vec{\hat\tabval}$.

Assume that~$\type(t)=\rem$. The case is slightly easier as the one above, and the remainder works similar.

Similarly, one can show the result for the remaining node with~$\type(t)=\join$, but define \PRIM-row solutions for two preceding child nodes of~$t$.
\end{proof}

We are now in the position to proof our theorem.

\begin{restatetheorem}[thm:primcorrectness]%
\begin{theorem}%
  The algorithm~$\dpa_\PRIM$ is correct. %
  More precisely, %
  the algorithm~$\dpa_\PRIM((\prog,\cdot),\TTT)$ returns
  $\PRIM$-TTD~$(T,\chi,\tau)$ such that we can decide consistency of~$\prog$ and even reconstruct the answer sets of~$\prog$:
  \begin{align*}
	&\mathcal{I}(\Ext_{\leq n}[\tau(n)])=
	\{\hat I \mid  \langle \hat I, \hat{\mathcal{P}}, \hat\sigma \rangle \text{ is a } \PRIM\text{-solution up to }n\}\\
	&=\{I \mid I \in \ta{\at(\prog)}, I \text{ is an answer set of }\prog\}.
  \end{align*}
\end{theorem}
\end{restatetheorem}
\begin{proof}[Proof (Idea)]
  By Lemma~\ref{lem:paspcorrect} we have soundness for every
  node~$t \in N$ and hence only valid rows as output of table
  algorithm~$\PRIM$ when traversing the tree decomposition in
  post-order up to the root~$n$.
  By Proposition~\ref{prop:hcfglobal} we then know that we can reconstruct answer sets
  given~\PRIM-solutions up to~$n$.
  In more detail, we proceed by means of induction. 
  For the induction base we only store~\PRIM-row solutions~$\vec\tabval\in\tau(t)$ at a certain node~$t$ starting at the leaves.
  For nodes~$t$ with~$\type(t)=\leaf$, obviously there is only the following (one)~\PRIM-row solution at~$t$: $\vec\tabval=\langle \emptyset, \emptyset, \langle \rangle\rangle$.
  
  Then, by Lemma~\ref{lem:paspcorrect} we establish the induction step, since algorithm~\PRIM only creates~\PRIM-row solutions at every node~$t$,
  assuming that it gets~\PRIM-row solutions at~$t_i$ for every child node~$t_i$ of~$t$.
  As a result, if there is no answer set of~$\prog$, the table~$\tau(n)$ is empty.
  On the other hand, if there is an answer set of~$\prog$, we obtain a~\PRIM-row solution~$\vec\tabval$ at root node~$n$, 
  for which by Definition~\ref{def:localhcf} a corresponding~\PRIM-solution~$\vec{\hat\tabval}$ up to~$n$ exists.
  Further, in the induction step we ensured that~\PRIM-solutions up to~$t$ for every~\PRIM-row solution at~$t$ for every node~$t\in N$ can be found that are compatible to~$\vec{\hat\tabval}$. In other words, by keeping track of corresponding origin~\PRIM-row solutions of~$\vec\tabval$ we can combine interpretation positions~$\mathcal{I}(\cdot)$ of rows by following origin rows top-down in order to reconstruct only valid answer set.

  Next, we establish completeness by induction starting from the
  root~$n$. Let therefore,
  $\hat\rho=\langle \hat I, \hat{\mathcal{P}}, \hat\sigma \rangle$ be
  the~\PRIM-solution up to node~$n$. If~$\hat\rho$ does not exist for
  node~$n$, there is by definition no answer set
  of~$\prog$. Otherwise, by Definition~\ref{def:localhcf}, we know
  that for the root~$n$ we can construct \PROJ-row solutions at~$n$ of
  the form~$\rho=\langle\emptyset, \emptyset, \langle \rangle\rangle$
  for~$\hat\rho$.  We already established the induction step in
  Lemma~\ref{lem:primcomplete} using~$\rho$ and~$\hat\rho$. As a
  consequence, we can reconstruct exactly \emph{all the answer sets}
  of~$\prog$ by following origin rows (see
  Definition of~$\orig$) back to the leaves and combining
  interpretation parts~$\mathcal{I}(\cdot)$, accordingly.
  Hence, we obtain some (corresponding) rows for every
  node~$t$. Finally, we stop at the leaves.

  In consequence, we have shown both soundness and completeness. As a
  result, Theorem~\ref{thm:primcorrectness} is sustains.
\end{proof}

\begin{corollary}\label{cor:primcorrectness}
  Algorithm~$\dpa_\PRIM((\prog,\cdot),\TTT)$ returns
  $\PRIM$-TTD~$(T,\chi,\tau)$ such that:
  \begin{align*}
    &\mathcal{I}(\PExt_{\leq t}[\tau(t)])\\
    &=\{\hat I \mid  \langle \hat I, \hat{\mathcal{P}}, \hat\sigma \rangle \text{ is a } \PRIM\text{-solution up to }t, \text{ there is answer set }\\
    &\quad\;\;\, I' \supseteq \hat I \text{ of } \prog \text{ such that }
      I' \subseteq I \cup (\at(\prog) \setminus \att{t})\}\\
    &=\{I \mid I \in \ta{\att{t}}, I \models \prog_{\leq t}, \text{ there is an answer set }\\
    &\quad\;\;\, I' \supseteq I \text{ of } \prog \text{ such that } I'\subseteq I \cup (\at(\prog) \setminus \att{t})\}.
  \end{align*}
\end{corollary}
\begin{proof}
The corollary follows from the proof of Theorem~\ref{thm:primcorrectness} applied up to node~$t$ and by considering only rows that are involved in reconstructing answer sets (see Definition~\ref{def:satext}).
\end{proof}

\subsection{Correctness of~$\mdpa{\AlgA}$: Omitted proofs}

In the following, we assume~$(\prog, P)$ to be an instance of~$\PASP$. Further, let~$\mathcal{T}=(T,\chi,\tau)$
be an~$\AlgA$-TTD of~$G_\prog$ where~$T=(N,\cdot,n)$, node~$t\in N$, and~$\rho\subseteq\tau(t)$.

\begin{definition}\label{def:asplocalsol}
Table algorithm $\AlgA$ is referred to as \emph{admissible}, if for each row $\vec{u_{t.i}}\in\tau(t)$ of any node~$t\in T$ the following holds:
  \begin{enumerate}
    \item $\mathcal{I}(\vec{\tabval_{t.i}}) \subseteq \chi(t)$
    \item For any $\langle\vec v_1, \ldots, \vec v_\ell\rangle \in \origse{\AlgA}(t,\vec{\tabval_{t.i}})$ where $1 \leq j \leq \ell$ and~$\children(t,T) = \langle t_1, \ldots, t_\ell \rangle$, we have $\mathcal{I}(\vec v_j) \cap \chi(t_j) \cap \chi(t) = \mathcal{I}(\vec u_{t.i}) \cap \chi(t_j) \cap \chi(t)$
    \item $\mathcal{I}(\PExt_{\leq t}[\tau(t)]) = \{I \mid I \in
    \ta{\att{t}}, I \models \prog_{\leq t}, \text{ there is an answer set } I \cup (\at(\prog) \setminus \att{t}) \supseteq  I' \supseteq I \text{ of } \prog\}$
    \item If~$t=n$ or~$\type(t)=\leaf$: $\Card{\local(t,\PExt_{\leq t}[\tau(t)])} \leq 1$
  \end{enumerate}
\end{definition}

Note that the last condition is not a hard restriction, since the bags of the leaf and root nodes of a tree decomposition are defined to be empty anyway. However, it rather serves as technical trick simplifying proofs.

\begin{observation}
Table algorithms~$\PRIM$ and~$\algo{PRIM}$ are admissible.
\end{observation}
\begin{proof}
  Obviously, Conditions 1, 2, and 4 hold by construction of the table algorithms and by properties auf tree decompositions. For condition 3, we have to check for correctness and completeness, which has been shown~\cite{FichteEtAl17a} for algorithm~$\algo{PRIM}$. For~$\PRIM$, see Theorem~\ref{thm:primcorrectness} and Corollary~\ref{cor:primcorrectness}.
\end{proof}

In the following, we assume that whenever~$\AlgA$ occurs, $\AlgA$ is an admissible table algorithm.

\begin{proposition}\label{prop:sat}
$\mathcal{I}(\PExt_{\leq n}[\tau(n)]) = \mathcal{I}(\Exts) = \{I \mid I \in
    \ta{\at(\prog)}, I \text{ is an answer }$ $\text{ set of } \prog\}.$
\end{proposition}
\begin{proof}
  Fill in Definition~\ref{def:asplocalsol} with root~$n$ of $\AlgA$-TTD ${\cal T}$.
\end{proof}

The following definition is key for the correctness proof, since later we show that these are equivalent with the result of~$\dpa_\PROJ$ using purged table mapping~$\nu$.

\begin{definition}\label{def:pmc}
  The \emph{projected answer sets count} $\pmc_{\leq t}(\rho)$ of
  $\rho$ below~$t$ is the size of the union over projected
  interpretations of the satisfiable extensions of~$\rho$ below~$t$,
  formally,
  $$\pmc_{\leq t}(\rho) \eqdef \Card{\bigcup_{\vec u\in\rho}
    \mathcal{I}_P(\PExt_{\leq t}(\{\vec u\}))}$$.

  The \emph{intersection projected answer sets count}
  $\ipmc_{\leq t}(\rho)$ of $\rho$ below~$t$ is the size of the
  intersection over projected interpretations of the satisfiable
  extensions of~$\rho$ below~$t$,~i.e.,
  $$\ipmc_{\leq t}(\rho) \eqdef \Card{\bigcap_{\vec u\in\rho}
    \mathcal{I}_P(\PExt_{\leq t}(\{\vec u\}))}$$.
\end{definition}

In the following, we state definitions required for the correctness
proofs of our algorithm \PROJ.  In the end, we only store rows that
are restricted to the bag content to maintain runtime bounds. 
We define the
content of our tables in two steps. First, we define the properties of
so-called \emph{$\PROJ$-solutions up to~$t$}. Second, we restrict
these solutions to~\emph{$\PROJ$-row solutions} at~$t$.

\begin{definition}\label{def:globalsol}
  Let~$\emptyset \subsetneq \rho \subseteq \tau(t)$ be a
  table with $\rho \in \subbuckets_P(\tau(t))$.
  We define a \emph{${\PROJ}$-solution up to~$t$} to be the sequence
  $\langle \hat {\rho}\rangle = \langle\PExt_{\leq t}(\rho)\rangle$.
\end{definition}

Before we present equivalence results between~$\ipmc_{\leq t}(\ldots)$
and the recursive version~$\ipmc(t, \ldots)$
used during the computation of
$\dpa_\PROJ$, recall that~$\ipmc_{\leq t}$ and~$\pmc_{\leq t}$
(Definition~\ref{def:pmc}) are key to compute the projected answer sets
count. The following corollary states that computing $\ipmc_{\leq n}$
at the root~$n$ actually suffices to compute~$\pmc_{\leq n}$, which is
in fact the projected answer sets count of the input program.

\begin{corollary}\label{cor:psat}
  \begin{align*}
    &\ipmc_{\leq n}(\local(n,\PExt_{\leq n}[\tau(n)]))\\
 =& \pmc_{\leq n}(\local(n,\PExt_{\leq n}[\tau(n)]))\\
    =& \Card{\mathcal{I}_P(\PExt_{\leq n}[\tau(n)])}\\
    =& \Card{\mathcal{I}_P(\Exts)}\\
    =& \,|\{J \cap P
       \mid J \in \ta{\at(\prog)},\\
    & J \text{ is an answer set of } \prog\}|
  \end{align*}
\end{corollary}
\begin{proof}
  The corollary immediately follows from Proposition~\ref{prop:sat}
  and since the cardinality of $\local(n,$ $\PExt_{\leq n}[\tau(n)])$ is at
  most one at root~$n$, by Definition~\ref{def:asplocalsol}.
\end{proof}

The following lemma establishes that the \PROJ-solutions up to
root~$n$ of a given tree decomposition solve the \PASP problem.

\begin{lemma}\label{lem:global}
  The
  value~$c = \sum_{\langle\hat{\rho}\rangle\text{ is a \PROJ-solution
      up to } n}\Card{\mathcal{I}_P(\hat{\rho})}$ corresponds to the
  projected answer sets count of~$\prog$ with respect to the set~$P$ of
  projection atoms.
\end{lemma}
\begin{proof}
  (``$\Longrightarrow$''): Assume
  that~$c = \sum_{\langle\hat{\rho}\rangle\text{ is a \PROJ-solution
      up to } n}\Card{\mathcal{I}_P(\hat {\rho})}$. Observe that there can be at
  most one projected solution up to~$n$ by Definition~\ref{def:asplocalsol}. %
  If~$c=0$, then $\tau(n)$ contains no rows. Hence, $\prog$ has no
  answer sets,~c.f., Proposition~\ref{prop:sat}, and obviously also no
  answer sets projected to~$P$. Consequently, $c$ is the projected answer sets
  count of~$\prog$.  
  If~$c>0$ we have by Corollary~\ref{cor:psat} that~$c$ is
  equivalent to the projected answer sets count of~$\prog$ with respect to~$P$.

  (``$\Longleftarrow$''): The proof proceeds similar to the only-if
  direction.
\end{proof}

\medskip %

In the following, we provide for a given node~$t$ and a given \PROJ-solution up to~$t$,
the definition of a \PROJ-row solution at~$t$.

\begin{definition}\label{def:loctab}~%
  Let %
  $\langle \hat{\rho} \rangle$ be a~$\PROJ$-solution up to~$t$. Then, we
  define the \emph{$\PROJ$-row solution at $t$} by
  $\langle \local(t,\hat{\rho}), \Card{\mathcal{I}_P(\hat{\rho})}\rangle$.
\end{definition}

\begin{observation}\label{obs:unique}
  Let $\langle \hat {\rho}\rangle$ be a \PROJ-solution up to a
  node~$t\in N$.  There is exactly one corresponding \PROJ-row
  solution
  $\langle \local(t,\hat{\rho}), \Card{\mathcal{I}_P(\hat{\rho})}\rangle$ at~$t$.

  Vice versa, let $\langle \rho, c\rangle$ at~$t$ be a \PROJ-row
  solution at~$t$ for some integer~$c$. Then, there is exactly one
  corresponding \PROJ-solution~$\langle\PExt_{\leq t}(\rho)\rangle$
  up to~$t$.
\end{observation}

We need to ensure that storing~$\PROJ$-row solutions at a
node~$t \in N$ suffices to solve the~\PASP problem, which is necessary
to obtain the runtime bounds as presented in
Corollary~\ref{cor:runtime}. For the root node~$n$, this is sufficient, shown in the following.

\begin{lemma}\label{lem:local}
  There is a
  \PROJ-row solution at the root~$n$ if and only if the projected
  answer sets count of~$\prog$ is larger than zero. Further, if there is a \PROJ-row solution~$\langle \rho, c\rangle$ at root~$n$, then~$c$ is the projected answer sets count of~$\prog$.
\end{lemma}
\begin{proof}%

  (``$\Longrightarrow$''): Let $\langle \rho, c\rangle$ be a
  \PROJ-row solution at root~$n$ where $\rho$ is an $\AlgA$-table and
  $c$ is a positive integer. Then, by Definition~\ref{def:loctab}
  there also exists a
  corresponding~$\PROJ$-solution~$\langle \hat{\rho} \rangle$ up
  to~$n$ such that $\rho = \local(n,\hat{\rho})$ and
  $c=\Card{\mathcal{I}_P(\hat{\rho})}$.
  Moreover, by Definition~\ref{def:asplocalsol}, we
  have~$\Card{\local(n,\PExt_{\leq n}[\tau(n)])}=1$.  
  Then, by Definition~\ref{def:globalsol},
  $\hat{\rho} = \PExt_{\leq n}[\tau(n)]$. By Corollary~\ref{cor:psat}, we
  have $c=\Card{\mathcal{I}_P(\PExt_{\leq n}[\tau(n)])}$ equals the projected answer sets count of~$\prog$.
  Finally, the claim follows.

  (``$\Longleftarrow$''): The proof proceeds similar to the only-if
  direction.
\end{proof}

Before we show that \PROJ-row solutions suffice, we require the following lemma.

\begin{observation}\label{obs:main_incl_excl}
  Let $n$ be a positive integer, $X = \{1, \ldots, n\}$, and $X_1$,
  $X_2$, $\ldots$, $X_n$ subsets of $X$.
  The number of elements in the intersection over all sets~$A_i$ is
    \[\Card{\bigcap_{i \in X} X_i} 
    = %
       \Bigg|\,\Card{\bigcup^n_{j = 1} X_j} %
                                         + \sum_{\emptyset \subsetneq I \subsetneq X} (-1)^{\Card{I}} 
                                              \Card{\bigcap_{i \in I} X_i}\,\Bigg|.\]
\end{observation}
\begin{proof}
  We take the well-known inclusion-exclusion
  principle~\cite{GrahamGrotschelLovasz95a} and rearrange the
  equation.
\end{proof}

\begin{lemma}\label{lem:main_incl_excl}
  Let $t\in N$ be a node of~$\TTT$
  with~$\children(t,T) = \langle t_1, \ldots, t_\ell \rangle$ and let
  $\langle\rho,\cdot\rangle$ be a~\PROJ-row solution at~$t$. Further, let~$\pi$ be a partial mapping of~$\pi'$ (finally returned by~$\dpa_\PROJ((\prog,P),\TTT)=(T,\chi,\pi')$), which maps nodes of the sub-tree~$T[t]$ rooted at~$t$ (excluding~$t$) to~$\PROJ$-tables.
  Then,
  \begin{enumerate}
  \item %
    $\ipmc(t,\rho,\langle\pi(t_1), \ldots,
    \pi(t_\ell)\rangle) = \ipmc_{\leq t}(\rho)$
  \item \smallskip%
    for $\type(t) \neq \leaf$:\\
    $\pmc(t,\rho,\langle\pi(t_1), \ldots,
    \pi(t_\ell)\rangle) = \pmc_{\leq t}(\rho)$.
  \end{enumerate}
\end{lemma}
\begin{proof}[Proof (Sketch)]
  We prove the statement by simultaneous induction.
  
  (``Induction Hypothesis''): Lemma~\ref{lem:main_incl_excl} holds for the nodes in~$\children(t,T)$ and also for node~$t$, but on strict subsets~$\varphi\subsetneq\rho$.
  (``Base Cases''): Let $\type(t) = \leaf$.
  Then by definition,
  $\ipmc(t,\{\langle \emptyset, \ldots\rangle\}, \langle \rangle) = \ipmc_{\leq t}(\{\langle\emptyset,\ldots\rangle\}) =
  1$.  
  Recall that for $\pmc$ the equivalence does not hold for leaves, but we use a node
  that has a node~$t'\in N$ with~$\type(t') = \leaf$ as child for the
  base case. Observe that by definition of a tree decomposition
  such a node~$t$ can have exactly one child.
  Then, we have that
  $\pmc(t,\rho,\langle\pi(t')\rangle) = \sum_{\emptyset
    \subsetneq O \subseteq {\origs(t,\rho)}} (-1)^{(\Card{O} - 1)}
  \cdot \sipmc(\langle \tau(t')\rangle, O) =
  \Card{\bigcup_{\vec u\in\rho} \mathcal{I}_P(\PExt_{\leq t}(\{\vec u\}))} =
  \pmc_{\leq t}(\rho) = 1$ where $\langle\rho,\cdot\rangle$ is
  a~\PROJ-row solution at~$t$.

  (``Induction Step''): We proceed by case distinction.

  Assume that $\type(t) = \intr$.
  Let $a \in (\chi(t) \setminus \chi(t'))$ be an introduced
  atom. We have two cases. Case (i) $a$ also belongs to
  $(\at(\prog) \setminus P)$,~i.e., $a$ is not a projection atom; and
  Case (ii) $a$ also belongs to $P$,~i.e., $a$ is a projection
  atom.
  Assume that we have Case~(i).
  Let~$\langle \rho, c \rangle$ be a \PROJ-row solution at~$t$ for
  some integer~$c$. As a consequence of admissible algorithm~$\AlgA$ (see Definition~\ref{def:asplocalsol})
  there can be many rows in the table~$\tau(t)$ for one row in
  the table~$\tau(t')$, more precisely,
  $\Card{\buckets_P(\rho)} = 1$.
  As a result,
  $\pmc_{\leq t}(\rho) = \pmc_{\leq t'}(\orig(t,\rho))$ by
  applying Observation~\ref{obs:unique}.
  We apply the inclusion-exclusion principle on every subset~$\varphi$ of
  the origins of~$\rho$ in the definition of~$\pmc$ and by induction
  hypothesis we know that
  $\ipmc(t',\varphi,\langle\pi(t')\rangle) = \ipmc_{\leq
    t'}(\varphi)$, therefore,
  $\sipmc(\pi(t'), \varphi) = \ipmc_{\leq t'}(\varphi)$.  This
  concludes Case~(i) for $\pmc$. The induction step for $\ipmc$ works
  similar %
  by applying
  Observation~\ref{obs:main_incl_excl} and comparing the corresponding
  \PROJ-solutions up to~$t$ or $t'$, respectively. 
  Further, for showing the lemma for~$\ipmc$, one has to additionally apply the hypothesis for node~$t$, but on strict subsets~$\emptyset\subsetneq\varphi\subsetneq\rho$ of~$\rho$.
  Assume that we have Case~(ii). We proceed similar as in Case~(i),
  since Case~(ii) is just a special case here, more precisely, we also
  have $\Card{\buckets_P(\rho)} = 1$ here.

  Assume that $\type(t) = \rem$. Let
  $a \in (\chi(t') \setminus \chi(t))$ be a removed atom. We have
  two cases. Case (i) $a$ also belongs to
  $(\at(\prog) \setminus P)$,~i.e., $a$ is not a projection atom; and
  Case (ii) $a$ also belongs to $P$,~i.e., $a$ is a projection
  atom.
  Assume that we have Case~(i).  Let~$\langle \rho, c \rangle$ be a
  \PROJ-row solution at~$t$ for some integer~$c$.
  As a consequence of admissible table algorithms~$\AlgA$ (see Definition~\ref{def:asplocalsol}) there can be many rows
  in the table~$\tau(t)$ for one row in the
  table~$\tau(t')$ (and vice-versa). Nonetheless we still have
  $\pmc_{\leq t}(\rho) = \pmc_{\leq t'}(\orig(t,\rho))$, because
  $a \notin P$ by applying Observation~\ref{obs:unique}.
  We apply the inclusion-exclusion principle on every subset~$\varphi$ of
  the origins of~$\rho$ in the definition of~$\pmc$ and by induction
  hypothesis we know that
  $\ipmc(t',\varphi,\langle\pi(t')\rangle) = \ipmc_{\leq
    t'}(\varphi)$, therefore,
  $\sipmc(\pi(t'), \varphi) = \ipmc_{\leq t'}(\varphi)$.  This
  concludes Case~(i) for $\pmc$. Again, the induction step for $\ipmc$
  works similar, but swapped.
  Assume that we have Case~(ii).
  Let~$\langle \rho, c \rangle$ be a \PROJ-row solution at~$t$ for
  some integer~$c$.
  Here we cannot ensure
  $\pmc_{\leq t}(\rho) = \pmc_{\leq t'}(\orig(t,\rho))$, since
  buckets fall together.  However, by applying
  Observation~\ref{obs:unique} we have
  $\pmc_{\leq t}(\rho) = \sum_{\varphi \in
    \buckets_P(\origs(t,\rho)_{(1)})} \pmc(t', \varphi, C) $ where the
  sequence~$C$ consists of the tables~$\pi(t'_i)$ of the children~$t'_i$ of~$t'$.
  For every~$\varphi \in \subbuckets_P(\origs(t,\rho)_{(1)})$ by
  induction hypothesis we know that
  $\ipmc(t',\varphi,\langle\pi(t')\rangle) = \ipmc_{\leq
    t'}(\varphi)$.
  Hence, we apply the inclusion-exclusion principle over all
  subsets~$\zeta$ of~$\varphi$ for all~$\varphi$ independently.  By
  construction
  $\sipmc(\pi(t'), \zeta) = \ipmc_{\leq t'}(\zeta)$.  Then,
  by construction
  $\pcnt(t,\rho, C') = \sum_{\emptyset \subsetneq O \subseteq
    {\origs(t,\rho)}} (-1)^{(\Card{O} - 1)} \cdot \sipmc(C', O) =
  \pmc_{\leq t}(\rho)$, where
  $C' = \langle \pi(t') \rangle$, since for the remaining
  terms $\sipmc(C', O)$ is simply zero, including cases where
  different buckets are involved.
  This concludes Case~(ii) for $\pmc$. Again, the induction step for
  $\ipmc$ works similar, but swapped by again applying
  Observation~\ref{obs:main_incl_excl}.

  Assume that $\type(t) = \join$. We proceed similar to the introduce
  case. However, we have two \PROJ-tables for the children of~$t$.
  Hence, we have to consider both sides when computing $\sipmc$
  (see Definition of~$\sipmc$). %
  There we consider the
  cross-product of two \AlgA-tables and we can also correctly apply
  the inclusion-exclusion principle on subsets of this cross-product,
  which we can do by simply multiplying $\sipmc$-values
  accordingly. The multiplication is closely related to the join case
  in table algorithm~\AlgA. For $\ipmc$ this does not apply, since the
  inclusion-exclusion principle is carried out at the node~$t$ and not
  for its children.

  Since we outlined all cases that can occur for node~$t$, this
  concludes the proof sketch.
\end{proof}

\begin{lemma}[Soundness]\label{lem:correct}
  Let $t\in N$ be a node of~$\TTT$
  with~$\children(t,T) = \langle t_1, \ldots, t_\ell \rangle$.
  Then, each row~$\langle \rho, c \rangle$ at node~$t$ constructed
  by table algorithm~$\PROJ$ is also a~\PROJ-row solution for
  node~$t$.
\end{lemma}
\begin{proof}[Proof (Idea)]
  Observe that Listing~\ref{fig:dpontd3} computes a row for each
  sub-bucket $\rho \in \subbuckets_P(\local(t,$ $\PExt_{\leq t}[\tau(t)]))$. The
  resulting row~$\langle\rho, c \rangle$ obtained by~$\ipmc$ is
  indeed a \PROJ-row solution for~$t$ according to
  Lemma~\ref{lem:main_incl_excl}.
\end{proof}

\begin{lemma}[Completeness]\label{lem:complete}
  Let~$t\in N$ be node of~$\TTT$ where
  $\type(t) \neq \leaf$ and~$\children(t,T) = \langle t_1, \ldots, t_\ell \rangle$. Given a
  \PROJ-row solution~$\langle \rho, c \rangle$ at node~$t$.
  There exists $\langle C_1, \ldots, C_\ell\rangle$ where $C_i$ is set
  of \PROJ-row solutions at~$t_i$ %
  such that
  $\rho \in \PROJ(t, \cdot, \tau(t), \cdot, P, \langle C_1, \ldots,
  C_\ell\rangle)$.
\end{lemma}
\begin{proof}[Proof (Idea)]
Since~$\langle\rho,c \rangle$ is a~\PROJ-row solution for~$t$, there is by Definition~\ref{def:loctab} a corresponding ~\PROJ-solution~$\langle\hat\rho\rangle$ up to~$t$ such that~$\local(t,\hat\rho) = \rho$. 

We proceed again by case distinction. Assume that~$\type(t)=\intr$. Then we define~$\hat{\rho'}\eqdef \{(t',\hat\varphi) \mid (t', \hat\varphi)\in \rho, t \neq t'\}$. Then, for each subset~$\emptyset\subsetneq\varphi\subseteq\local(t',\hat{\rho'})$, we define~$\langle \varphi, \Card{\mathcal{I}_P(\PExt_{\leq t}(\varphi))}\rangle$ in accordance with Definition~\ref{def:loctab}. By Observation~\ref{obs:unique}, we have that~$\langle \varphi, \Card{\mathcal{I}_P(\PExt_{\leq t}(\varphi))}\rangle$ is an \AlgA-row solution at node~$t'$. 
Since we defined the~\PROJ-row solutions for~$t'$ for all the respective \PROJ-solutions up to~$t'$, we encountered every~\PROJ-row solution for~$t'$ that is required for deriving~$\langle \rho, c\rangle$ via~\PROJ (c.f., Definitions of~$\ipmc$ and of~$\pmc$). %

Assume that~$\type(t)=\rem$. The case is slightly easier as the one
above. We do not need to define a~\PROJ-row solution for~$t'$ for all
subsets~$\varphi$, since we only have to consider subsets~$\varphi$ here,
with~$\Card{\buckets_P(\varphi)}=1$. The remainder works similar.

Similarly, one can show the result for the remaining node with~$\type(t)=\join$, but define \PROJ-row solutions for two preceding child nodes of~$t$.
\end{proof}

We are now in the position to proof our theorem.

\begin{theorem}\label{thm:correctness}
  The algorithm~$\dpa_\PROJ$ is correct. %
  More precisely, %
  the algorithm~$\dpa_\PROJ((\prog,P),\TTT)$ returns
  $\PROJ$-TTD~$(T,\chi,\pi$) such that $c=\sipmc(\pi(n), \cdot)$
  is the projected answer sets count of~$\prog$ with respect to the set~$P$ of
  projection atoms.
\end{theorem}
\begin{proof}
  By Lemma~\ref{lem:correct} we have soundness for every
  node~$t \in N$ and hence only valid rows as output of table
  algorithm~$\PROJ$ when traversing the tree decomposition in
  post-order up to the root~$n$.
  By Lemma~\ref{lem:local} we know that the projected answer sets count~$c$
  of~$\prog$ is larger than zero if and only if there exists a
  certain~\PROJ-row solution for~$n$.
  This~\PROJ-row solution at node~$n$ is of the
  form~$\langle \{\langle\emptyset, \ldots\rangle\} ,c\rangle$. If
  there is no \PROJ-row solution at node~$n$,
  then~$\tau(n)=\emptyset$ since the table algorithm~$\AlgA$
  is admissible (c.f., Proposition~\ref{prop:sat}). Consequently, we have
  $c=0$. Therefore, $c=\sipmc(\pi(n), \cdot)$ is the
  projected answer sets count of~$\prog$ with respect to~$P$ in both cases.

  Next, we establish completeness by induction starting from the
  root~$n$. Let therefore, $\langle \hat\rho \rangle$ be
  the~\PROJ-solution up to node~$n$, where for each row
  in~$\vec u\in \hat\rho$, $\mathcal{I}(\vec u)$ corresponds to an
  answer set of~$\prog$.  By Definition~\ref{def:loctab}, we know that
  for the root~$n$ we can construct a \PROJ-row solution at~$n$ of the
  form~$\langle \{\langle\emptyset, \ldots\rangle\} ,c\rangle$
  for~$\hat\rho$.  We already established the induction step in
  Lemma~\ref{lem:complete}.
  Hence, we obtain some (corresponding) rows for every
  node~$t$. Finally, we stop at the leaves.

  In consequence, we have shown both soundness and completeness. As a
  result, Theorem~\ref{thm:correctness} is sustains.
\end{proof}

\begin{corollary}\label{cor:correctness}
  The algorithm $\mdpa{\AlgA}$ is correct and outputs for any instance
  of \PASP its projected answer sets count.
\end{corollary}
\begin{proof}
  The result follows immediately, since~$\mdpa{\AlgA}$ consists of two
  dynamic programming passes~$\dpa_\AlgA$, a purging step, and~$\dpa_\PROJ$. For the
  soundness and completeness of~$\dpa_\algo{PRIM}$ we refer to other
  sources~\cite{FichteEtAl17a}. By Proposition~\ref{prop:sat}, the
  ``purging'' step does neither destroy soundness nor completeness
  of~$\dpa_\AlgA$.
\end{proof}

\begin{restateproposition}[prop:phcworks]
\begin{proposition}
  The algorithm $\mdpa{\PRIM}$ is correct and outputs for any instance
  of \PASP its projected answer sets count.
\end{proposition}
\end{restateproposition}
\begin{proof}
This is a direct consequence of Corollary~\ref{cor:correctness}.
\end{proof}

\begin{restateproposition}[prop:disjworks]
\begin{proposition}
  The algorithm $\mdpa{\algo{PRIM}}$ is correct and outputs for any instance
  of \PDASP its projected answer sets count.
\end{proposition}
\end{restateproposition}
\begin{proof}
This is a direct consequence of Corollary~\ref{cor:correctness}.
\end{proof}

\longerversion{
\section{Correctness of QBF lower bound}

\begin{definition}[\cite{MarxMitsou16}]
  Given graph~$G=(V,E)$, integers~$i,j,r$ where~$i\leq j$ and total list-capacity function~$f: V \rightarrow \{i,\ldots,j\}$. Then an instance~$(G,r,f)$ of~$(i,j)$-Choosability Deletion asks for the existence of a set of vertices~$V'\subseteq V$ with~$\Card{V'}\leq r$ and~$V_1\eqdef V\setminus V'$, such that for all assignments~$\mathcal{L}: V_1 \rightarrow 2^{\{i,\ldots,j\}}$ with $\Card{\mathcal{L}(v)} = f(v)$ for all~$v\in V_1$, there is a coloring~$c: V_1 \rightarrow \mathcal{L}(v)$ such that for every edge~$(u,v)\in E\setminus (V_1\times V_1): c(u) \neq c(v)$.
\end{definition}

\begin{proposition}[\cite{MarxMitsou16}]
  Instances~$(G,r,f)$ of~$(1,4)$-Choosability Deletion where~$k$
  is the treewidth of~$G$, cannot be
  solved in time~${2^{2^{2^{o(k)}}}}\cdot \CCard{G}^{o(k)}$.
\end{proposition}

\begin{restateproposition}[prop:lampis3]
\begin{proposition}
  QBFs of the form $\forall V_1.\exists V_2.\forall V_3. E$ where~$k$
  is the treewidth of the primal graph of DNF formula~$E$, cannot be
  solved in time~${2^{2^{2^{o(k)}}}}\cdot \CCard{E}^{o(k)}$.
\end{proposition}
\end{restateproposition}
\begin{proof}
We proof the result by reducing from the problem~$(1,4)$-Choosability Deletion,
\end{proof}}

}

\end{document}